\documentclass[journal]{IEEEtran}
\usepackage{amsfonts}
\usepackage{mathtools}
\usepackage{amsmath}
\usepackage{amsthm}
\usepackage{adjustbox}
\usepackage{graphicx}
\usepackage{caption}
\usepackage{subcaption}
\usepackage{float, color}

\graphicspath{ {image/}}

\usepackage{amsmath}

\newcommand{\ProjT}{\mathbf{P}_{\mathbf{u}^\perp}^T} 
\newcommand{\Map}{\mathcal{P}_\mathbf{u}}
\newcommand{\Blur}{\mathcal{P}_{\mathbf{u},\tau}}
\newtheorem{theorem}{Theorem}
\newtheorem{corollary}{Corollary}

\ifCLASSINFOpdf
\else
\fi
\begin{document}

\onecolumn{\large{
This work has been submitted to the IEEE for possible publication. Copyright may be transferred without notice, after which this version may no longer be accessible.
\copyright 2019 IEEE. Personal use of this material is permitted. Permission from IEEE must be obtained for all other uses, including reprinting/republishing this material for advertising or promotional purposes, collecting new collected works for resale or redistribution to servers or lists, or reuse of any copyrighted component of this work in other works.}}

\twocolumn

%
\title{A Convolutional Forward and Back-Projection Model for Fan-Beam Geometry}
%
%
%

\author{Kai Zhang,~\IEEEmembership{Student Member,~IEEE,}
		Alireza Entezari,~\IEEEmembership{Senior Member, ~IEEE,}

\thanks{This project was funded in part by NSF under grant IIS-1617101.}
\thanks{Kai Zhang is with the Department of Computer and Information Science and Engineering (CISE), University of Florida, Gainesville, FL, 32611 USA (e-mail: zhangkai6@ufl.edu).}

\thanks{Alireza Entezari is with the Department of Computer and Information Science and Engineering (CISE), University of Florida, Gainesville, FL, 32611 USA (e-mail: entezari@ufl.edu).}

}

\maketitle

\begin{abstract}
Iterative methods for tomographic image reconstruction have great potential for enabling high quality imaging from low-dose projection data. The computational burden of iterative reconstruction algorithms, however, has been an impediment in their adoption in practical CT reconstruction problems. We present an approach for highly efficient and accurate computation of forward model for image reconstruction in fan-beam geometry in X-ray CT. The efficiency of computations makes this approach suitable for large-scale optimization algorithms with on-the-fly, memory-less, computations of the forward and back-projection. Our experiments demonstrate the improvements in accuracy as well as efficiency of our model, specifically for first-order box splines (i.e., pixel-basis) compared to recently developed methods for this purpose, namely Look-up Table-based Ray Integration (LTRI) and Separable Footprints (SF) in 2-D.
\end{abstract}


%
\IEEEpeerreviewmaketitle

\section{Introduction}
A wide range of imaging modalities use tomographic reconstruction
algorithms to form 2-D and 3-D images from their projection data. While
the classical Filtered Back Projection (FBP) algorithm, and its
variants, are widely used in practice \cite{pan2009commercial}, iterative
reconstruction algorithms hold great potential to enable high-quality
imaging from limited projection data (e.g., few views, limited-view, low-dose) and reduce exposure to radiation. Developments, over last several
decades, in this area often formulate the image reconstruction as an
(ill-posed) inverse problem where a regularized solution is found by an
iterative optimization algorithm. Several aspects of iterative
reconstruction algorithms \cite{gilbert1972iterative, andersen1984simultaneous, andersen1989algebraic, sidky2008image} overlap with active areas of research in
solving these optimization problems efficiently as well as image modeling
with regularization (e.g., sparsity-based, network-based) that enhance
the quality of recovered image from limited data.

An important issue in the performance of iterative reconstruction
algorithms is the discretization, and more generally the representation,
of images. The expansion of an image in a basis allows for derivation of
a {\em forward model}, $\mathbf{A}$, that relates the image coefficients $\mathbf c$ to the measurements $\mathbf y$ in the projection domain by a linear system: $\mathbf y = \mathbf{A} \mathbf c$. The entries of the forward model are computed from the contributions of
each basis function to each measurement in the projection domain.
Specifically, the contribution of a basis function is computed by (1)
integrating the basis function along an  incident ray to form its 
{\em footprint} and (2) integrating these footprints across a detector
cell in the projection domain (often called the {\em detector blur}).
Common choices for expansion are the pixel- and voxel-basis that provide
a piecewise-constant model for image representation. Kaiser-Bessel \cite{lewitt1990multidimensional}
functions have also been considered as a smooth basis for image
representation. 

Given a finite set of measurements $\mathbf y \in \mathbb{R}^m$ and a choice for basis function, one can setup linear systems for images at various resolutions (i.e., discretizations) to be reconstructed from $\mathbf y$. For an image resolution with $N$ elements, characterized by coefficients $\mathbf c \in \mathbb{R}^N$, the forward model is constructed from the footprint of scaled basis functions. To reduce the degrees of freedom in the inverse problem, one wants to build the forward model at the coarsest possible resolution for image discretization; however, doing so limits our ability to resolve for the features of interest in the image. The ability of (a space spanned by translates of) a basis function to approximate images with the smallest resolution is often characterized by approximation order -- scales of a basis function with higher approximation order provide discretizations that approximate the underlying image more accurately, compared to the scales of a basis function with a lower approximation order. While the choice of pixel-basis provides a partition of unity and hence a first-order method for approximation, Kaiser-Bessel  functions require filtering operations\cite{nilchian2015optimized} to achieve first-order approximation. A different class of basis functions, called box splines, have been proposed in \cite{entezari2012box} for image discretization in the context of tomographic reconstruction. The first-order box splines coincide with the pixel-basis (in 2-D) and voxel-basis in (3-D), but higher order box splines allow for higher orders of approximation. 

Besides approximation order, the effectiveness of basis functions for tomographic reconstruction also depends on the accuracy and efficiency of calculating integrals involved in the footprint and detector blur. The pixel-basis and Kaiser-Bessel functions as well as box splines have closed-form Radon transforms and their footprints can be computed analytically. However, computation of detector blur is more challenging, and several approaches have been proposed for approximating the underlying integrals, such as strip-integral method\cite{lo1988strip}, natural pixel decomposition \cite{byonocore1981natural}, Fourier-based methods \cite{tabei1992backprojection}, distance-driven \cite{de2002distance} approximation, and separable footprints (SF)~\cite{long20103d}. More recently a look-up table-based integration (LTRI) approach was proposed~\cite{ha2018look} that provides speedups compared to the SF method at the cost of further errors in approximating the integrals. 

In this paper, we demonstrate that the detector blur can be calculated with a convolution -- in the continuous domain -- using box spline methodology; we also demonstrate a practical approach for computing these projections in fan-beam geometry using box spline evaluation algorithms. The convolutional approach leads to efficient computations that are exact in parallel geometry and highly accurate in fan-beam geometry. While our method encompasses both flat and arc detectors, we will discuss flat detectors, since the extension to the arc geometry is easily obtained as a special case.
Specific contributions of this paper include:
\begin{itemize}
\item Derivation of fast forward and back-projection operators for exact computation of footprint integral in the fan-beam geometry.
\item Derivation of the accurate detector blur computation in both parallel and fan-beam geometry.
\end{itemize}

\section{X-ray optics Model}\label{Fan Beam Projection Geometry}
\subsection{Fan Beam Geometry}
To specify the fan-beam geometry in the general $2$-dimensional configuration, 
let $\mathbf{u} \in S$ denote viewing direction and point $\mathbf{o} \in 
\mathbb{R}^2$ as rotation center (or origin). A point  $\mathbf{p} = D_{po}\mathbf{u}$ is the location of projector, where $D_{po} \in \mathbb{R}^+$ is the (unsigned) distance from projector to rotation center.
A hyperplane $\mathbf{u}^\perp$ orthogonal to the direction $\mathbf{u}$ denotes the detector plane and point $\mathbf{d} \in \mathbf{u}^\perp$ is the center of the detector plane. Let $\ProjT$ be a matrix whose columns span the hyperplane $\mathbf{u}^\perp$
,  $\mathbf{x} \in \mathbb{R}^2$ be the spatial coordinates of the input function (image domain), and the parameterized form is
$\mathbf{x} =\mathbf{p} + t\mathbf{v}(s)$,
where ${s} \in \mathbb{R}$ is the coordinate on the 1-D detector plane (sinogram domain) and $\mathbf v(s)$ is the unit direction from $\mathbf p$ to $\ProjT s$, which can be calculated by 
$\frac{\mathbf{p} - \ProjT s}{ \|\mathbf{p} - \ProjT s \| }$. The detector-rotation center distance is
$D_{so}$ and projector-detector distance is $D_{ps}$. The geometry is 
illustrated in Fig. \ref{fig:fanBeamGeometry}.
\begin{figure}[ht]
\centering
{\includegraphics[width=0.8\columnwidth]{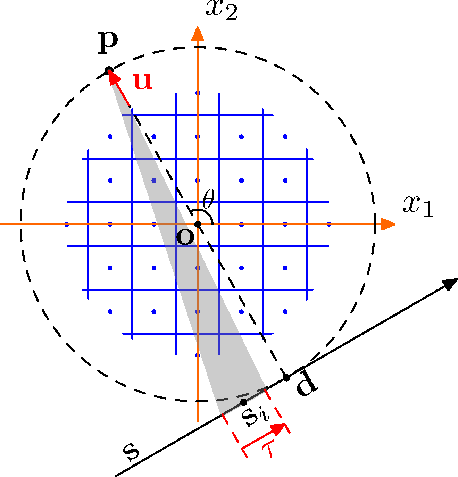}}
\caption{Fan-beam X-ray CT system, a discretized model}\label{fig:fanBeamGeometry}
\end{figure}

\subsection{Analytical Model}
In geometric tomography $f$ is often the indicator function of a convex polytope and in biomedical \cite{brahme2014comprehensive}, scientific imaging \cite{rebollo2013sparse}, and industrial applications \cite{natterer1986mathematics} $f$ is the relative radiodensity modeling material's attenuation index as the X-ray passes through the object (often described by linearization of the Lambert-Beer law \cite{long20103d}). When considering that the projector is ideal point source, the 2-dimensional fan-beam X-ray transform $\mathcal{P}$ maps $f(\mathbf x)$, 
$\mathbf x \in \mathbb{R}^2$, into the set of its \emph{line integrals} to form the projection:

\begin{equation}\label{eq:xRayIntegration}
\mathcal{P}\{f\}(s,\mathbf{u}) = \int_{0}^{\infty} f(\mathbf{p} +  t\mathbf{v}(s)) {\rm d}t.
\end{equation} 
We denote above map as $\Map\{f\}$ for short. 

In a simple model of forward projection, one can do point-sampling on the projected function $\Map \{f\}$, 
 whereas in more realistic modeling of the transform model the projections are integrated across a detector cell with a finite bin width ${\tau} \in \mathbb{R}$ such that equation
(\ref{eq:xRayIntegration}) can be modified to 
\begin{equation}\label{eq:radonBoxBlur}
\begin{split}
\mathcal{P}_{\mathbf{u}, {\tau}}\{f\}({s}) = & \int_{{s} - \frac{\boldsymbol{\tau}}{2} }^{ {s} + \frac{\boldsymbol{\tau}}{2} } h_\tau(s)\int_0^\infty f(\mathbf{p} +  t\mathbf{v}(s)){\rm d}t {\rm d}s
\end{split}
\end{equation}
where $h_\tau$ is the detector blur function over the support of bin width ${\tau}$. Since the detectors are usually uniformly placed on the projection
domain, the blur function is often modeled as shift-invariant function \cite{herman2009fundamentals, lu2009computable}. Because the detector sensitivity is often modeled as a constant function over the detector cell or with a drop-off at the cell boundary \cite{lu2009computable, garcia2009study}, the analytical model can be simplified\cite{yu2012finite}:
\begin{equation}\label{eq:anlyticalProjection}
\begin{split}
\mathcal{P}_{\mathbf{u}, {\tau}}\{f\}({s}) &= \int_{s-\frac{\tau}{2}}^{s+\frac{\tau}{2}}\int_{0}^{\infty} f(\mathbf{p} +  t\mathbf{v}(s)){\rm d}t {\rm d}s\\
&= \int_{\Omega_s}\frac{f(\mathbf{x})}{\gamma(s,\tau)\|\mathbf{x-p}\|} {\rm d}\mathbf{x},
\end{split}
\end{equation}
where $\Omega_s$ is the detector source-bin triangle shown as the gray area in
Fig. \ref{fig:fanBeamGeometry} and $\gamma(s, \tau)$ is the fan-beam angle at coordinate $s$ with bin width $\tau$.

Due to the linearity of integration, X-ray transform pseudo-commutes with the 
translation:
\begin{equation}\label{eq:translationProperty}
\begin{split}
\Map\{f(\cdot - \mathbf{x}_0)\}( s) &= (\Map f)(s - \mathbf{P(x}_0))\\
\Blur\{f(\cdot - \mathbf{x}_0)\}( s) &= (\Blur f)(s - \mathbf{P(x}_0)),
\end{split}
\end{equation}
where the operator $\mathbf{P}$ represents the \emph{perspective 
projection} with perspective division. As the geometric projection of a point from image domain onto sinogram domain is not orthogonal in fan-beam geometry (e.g., $\mathbf{ps}_i$ is not perpendicular to detector plane in Fig. \ref{fig:fanBeamGeometry}.), the perspective distortion achieved 
by perspective division is necessary to correctly reflect the relation between distance of a pixel to projector and its location on the projection plane. 


\subsection{Discretized Model}
Discretization or a finite-dimensional approximation of a continuous-domain signal (image) $f$ is an important issue in signal processing. With an $N$-dimensional model for approximation (see Fig. \ref{fig:fanBeamGeometry}), an expansion in a basis set of the form
\begin{equation}\label{eq:modelofInputFunction}
f_N(\mathbf x) = \sum_{n=1}^Nc_n \varphi(\mathbf{x-k}_n),
\end{equation}
allows us to derive a discretized forward model in tomography. Here $\varphi$ is a basis function, and $c_n$ is the expansion coefficient corresponding to the $n^{\rm th}$ grid point $\mathbf k_n$. The combination of (\ref{eq:xRayIntegration}), (\ref{eq:radonBoxBlur}), and this expansion, together with the translation property (\ref{eq:translationProperty}), provides the discretized forward models:
\begin{equation}\label{eq:modelofXrayIntegration}
\begin{split}
\Map\{f_N\}( s)  &= \sum_{n=1}^N c_n\Map\{\varphi\}( s - \mathbf{P(k}_n)) \\
\Blur\{f_N\}( s)  &= \sum_{n=1}^N c_n\Blur\{\varphi\}( s - \mathbf{P(k}_n)).
\end{split}
\end{equation}
The line integral of the basis function, $\Map\{\varphi\}$, is called \emph{footprint} of the basis and $\Blur\{\varphi\}$ is \emph{detector blur}.
{Equations (\ref{eq:modelofXrayIntegration}) show that the exact X-ray transform of the discretized model can be modeled by linear combination of the integral of the basis function.

In 2-D, the commonest choice of basis functions for image representation is the indicator function of pixels (aka the pixel-basis). With pixel-basis and a constant function to model the detector blur, a geometric approach for the combined footprint and detector blur computations can be derived~\cite{yu2012finite} as:
\begin{equation}\label{eq:anlyticalProjectionDiscrete}
\Blur\{f_N\}( s) \approx \sum_{{n} =1}^{N} c_n \frac{A_n}{\gamma(s, \tau)\|\mathbf k_n- \mathbf p\|},
\end{equation}
where $A_{n}$ is the intersection area of the detector source-bin triangle with the pixel located at $n^{\rm th}$ grid $\mathbf{k}_n$.
This intersection area can be computed exactly by the Gauss's area formula:
\begin{equation}\label{eq:groundTruth}
A_{n} = \frac{1}{2!}  \sum_{i=1}^{M-2} \Bigl| ( \mathbf{v}_i - \mathbf{v}_{i+1}) \times (\mathbf{v}_M - \mathbf{v}_{i+1}) \Bigr|,
\end{equation}
where $\mathbf{v}_i$ are the anticlockwise sorted vertices of a convex polygon obtained by intersection of detector source-bin triangle with a pixel, which can be found algorithmically (e.g., by applying Sutherland-Hodgman algorithm \cite{sutherland1974reentrant}), and $M(\geq 3)$ denotes the total number of intersections. 

{Searching for the intersection vertices is computationally expensive and hard to be parallelized due to the sequential nature of polygon clipping and sorting algorithm. The LTRI method~\cite{ha2018look} provides a solution to speedup these computations by using a look-up table to store a finite set of precomputed intersection areas and interpolating them for each combination of viewing angle, pixel and detector bin. This interpolation naturally introduces an extra source of error that can be controlled, to some extent, by increasing the resolution of the look-up table, but the data transfer involved in accessing a large look up table slows down GPU computations. In addition, a constant approximation of $\|\mathbf{x - p}\|$ by $\|\mathbf{k-p}\|$ degrades the accuracy in approximating the analytical model (\ref{eq:anlyticalProjection}).  Furthermore, this approximation is not flexible for higher-order cases because the intersection area needs to be generalized by with geometric integration when higher-order basis functions are used and this renders the volume computation infeasible.}

\subsection{Roadmap} Our method models the footprint, $\Map\{ \varphi\}$, in the continuous sinogram domain using box spline calculus (instead of approximating the optics integrals through intersection areas). This formulation allows us to leverage an exact closed-form formula in both parallel and fan-beam geometries. The detector blur, $\Blur \{\varphi\}$, calculated by integral of $\Map\{ \varphi\}$ within a detector bin, in our framework, also has an exact closed-form expression for parallel geometry. In fan-beam geometry, although this integral can be computed exactly, it has a prohibitively large computational cost. While we use that exact integration as the $\emph{reference projector}$, we also introduce the concept of \emph{effective blur} for efficient approximation of the detector blur. The highly efficient computation of detector blur, discussed in section \ref{sec:fanbeamXray}, has a closed-form solution in box spline methodology which leads to highly accurate and efficient computations that are the main results of this paper.

%


\section{Discretization in Box Spline Basis}
Box splines are piecewise polynomial functions that can be used as basis functions for approximation in discrete-continuous modeling. The pixel-basis coincides with the first order box spline in 2-D and higher order box splines can be considered as more general choices for discretization of images. However, since the pixel-basis is the most commonly-used choice in CT, we will view the pixel-basis using the terminology of box splines. While viewing pixel-basis as a box spline may appear as a complication, benefits of this formulation becomes apparent once we establish that footprint and detector blur integrals of the pixel-basis result in higher-order box splines. The resulting higher-order box splines allow us to efficiently and accurately model these operations that are essential to forward and back-projection.

\subsection{Box Spline Review}
Box splines generalize B-splines to multivariate setting where they include tensor-product B-splines as a special case, but are generally non-separable functions. A box spline is a smooth piecewise polynomial function, $M: \mathbb{R}^d \mapsto \mathbb{R}$, that is associated with $N$ vectors in $\mathbb{R}^d$~\cite{deboor1993box}. The simplest box spline (i.e., $N=d$) can be constructed in $\mathbb{R}^d$ as the indicator function of the $d$-dimensional hypercube that is the pixel- and voxel-basis function when d = 2 and 3. Box splines have a convolutional construction that is essential to our derivation of footprint and detector blur.

An {\em elementary} box spline, $M_{\boldsymbol{\xi}}$, associated with a vector
$\boldsymbol{\xi} \in \mathbb{R}^d$ can be thought of as the indicator function of the set $\{t\boldsymbol{\xi} | 0 \leq t < 1\}$, and is formally defined as a Dirac-line distribution (generalized function) by its directional ``moving-average'' action on a
function $f(\mathbf{x})$ in $\mathbb{R}^d$: 
$(M_{\boldsymbol{\xi}} * f)(\mathbf{x}) = \int_0^1 f(\mathbf{x} - t \boldsymbol{\xi}){\rm d}t$. Given a set of $N \ge d$ directions, arranged in columns, as
$\boldsymbol{\Xi} := [\boldsymbol{\xi}_1, \boldsymbol{\xi}_2, \cdots,
\boldsymbol{\xi}_N]$, the associated box spline can be constructed by:
\begin{equation}
M_{\boldsymbol{\Xi}}(\mathbf{x}) = (M_{\boldsymbol{\xi}_1} * 
\cdots * M_{\boldsymbol{\xi}_N})(\mathbf{x}),
\end{equation}
and this is illustrated in 2D (i.e., $d=2$) in Fig. \ref{fig:boxsplinesConvolution}.
\begin{figure}[ht]
\centering
{\includegraphics[width=1.0\columnwidth]{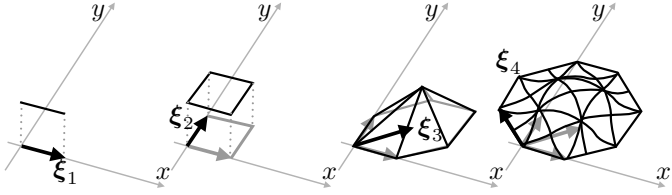}}
\caption{From left to right, the figures show that the Box Splines can be constructed by convolving an elementary Box Splines in specific direction.}\label{fig:boxsplinesConvolution}
\end{figure}
When the directions are orthogonal, $M_{\boldsymbol{\Xi}}$ is a tensor-product B-spline with the repeats of a direction elevating the degree of the B-spline.
\subsection{X-ray projection in Parallel Geometry}
Previous work~\cite{entezari2012box} has demonstrated that in \emph{parallel} geometry the projection (e.g., X-ray and Radon transforms) of a box spline $M_{\boldsymbol{\Xi}}$ is another box spline, $M_{\boldsymbol{Z}}$ (i.e., in the sinogram domain), where $\boldsymbol{Z}$ is the \emph{geometric} projection of the original directions in $\boldsymbol{\Xi}$. 
Let $\mathcal{R}_\mathbf{u}$ denote the X-ray projection for a viewing
direction specified by vector $\mathbf{u}$, then we have~\cite{entezari2012box}:
\begin{equation}\label{eq:projectionBoxSplineParallel}
\mathcal{R}_\mathbf{u}\{M_{\boldsymbol{\Xi}}\}({s}) = M_{\boldsymbol{Z}}({s})=
(M_{\mathbf{R}_{\mathbf{u}^\perp}\boldsymbol{\xi}_1} * \cdots *
 M_{\mathbf{R}_{\mathbf{u}^\perp} \boldsymbol{\xi}_N})({s}),
\end{equation}
where $\mathbf{R}_{\mathbf{u}^\perp}$ is the transformation matrix that geometrically projects the coordinates of image domain onto detector plane perpendicular to 
$\mathbf u$. Using this notation, we have $\boldsymbol{Z} := \mathbf{R}_{\mathbf{u}^\perp} \boldsymbol{\Xi}$.
Fig. \ref{fig:boxSplineProjConv} shows a box spline in $\mathbb{R}^2$ specified by directions $\boldsymbol{\Xi} = [\boldsymbol{\xi}_1, \boldsymbol{\xi}_2]$,
which coincides with a pixel-basis, when projected to the sinogram domain is
a univariate box spline in $\mathbb{R}$ with 2 directions that is the convolution of 2 elementary box splines. In this
example, the matrix $\mathbf{R}_{\mathbf{u}^\perp} = [\sin(\theta), 
-\cos{\theta}]$ specifies the direction of the projection plane, where 
$\theta$ is the viewing angle, and $\boldsymbol{\xi}_1 = 
\begin{bmatrix}
1\\
0
\end{bmatrix}
$, and $\boldsymbol{\xi}_2 = 
\begin{bmatrix}
0\\
1
\end{bmatrix}
$.
\begin{figure}[ht]
\centering
{\includegraphics[width=0.9\columnwidth]{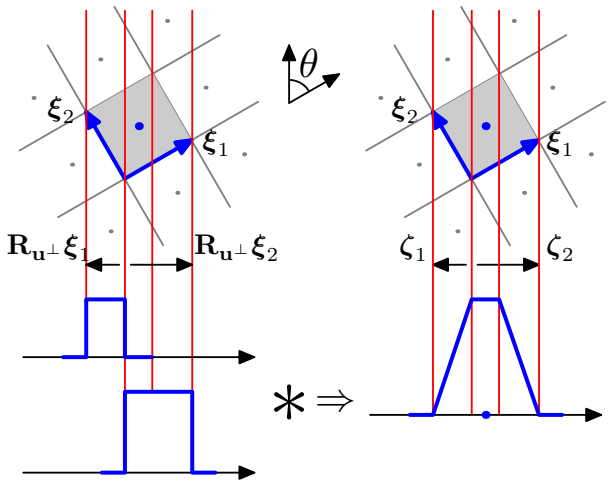}}
\caption{Projection of pixel basis as convolution. }\label{fig:boxSplineProjConv}
\end{figure}
\subsection{Detector Blur}
Since the detector sensitivity is often modeled
as a constant function, the blur function $h_{{\tau}}$ can be
modeled as a box spline in an elementary box spline so that
$h_{{\tau}} = M_{{\tau}}$. 
With this model, we can derive\footnote{A preliminary version of this manuscript appeared in conference proceedings in ISBI~\cite{Zhang2019box}.} the result of detector blur in parallel-beam geometry.
\begin{theorem}\label{theory:detectorBlurParallel}
The parallel-beam projection of pixel-basis with detector blur is a 3-direction box spline.
\end{theorem}
\begin{proof}
Because detector blur is the convolution of blur function $h_\tau$ with footprint $\mathcal{R}_\mathbf{u}\{M_{\boldsymbol{\Xi}}\}$.
By applying the convolutional construction of box spline, the detector blur can be computed by a single evaluation of a box spline:
\begin{equation*}
\begin{split}
\mathcal{R}_{\mathbf{u}, {\tau}}\{M_{\boldsymbol{\Xi}}\} ({s})= 
\mathcal{R}_{\mathbf{u}}\{M_{\boldsymbol{\Xi}}\} * h_{{\tau} }({s}) = M_{[\mathbf{R}_{\mathbf{u}^\perp}\boldsymbol{\Xi}, {\tau}]}({s}).
\end{split}
\end{equation*}
We also note that as the detector cells lie on the detector plane, it's naturally perpendicular to 
the direction of projection such that $\mathbf{R}_{\mathbf{u}^\perp}\boldsymbol{\tau} = \tau$ and $[\mathbf{R}_{\mathbf{u}^\perp}\boldsymbol{\Xi}, {\tau}] = \mathbf{R}_{\mathbf{u}^\perp}[\boldsymbol{\Xi}, \boldsymbol{\tau}]$, where $\boldsymbol{\tau}$ is a vector in image domain parallel to detector plane. Hence, it indicates that the detector blur can be modeled in image domain as convolution of basis function with an elementary box spline (see Fig. \ref{fig:boxSplineConvEff}):
\begin{equation}\label{eq:boxspineProjectionBinBlurParallel}
\begin{split}
\mathcal{R}_{\mathbf{u}, {\tau}}\{M_{\boldsymbol{\Xi}}\}( s)  = \mathcal{R}_{\mathbf{u}}\{M_{[\boldsymbol{\Xi, \tau}]}\}( s).
\end{split}
\end{equation}
\end{proof}

\begin{figure}[ht]
\centering
{\includegraphics[width=0.65\columnwidth]{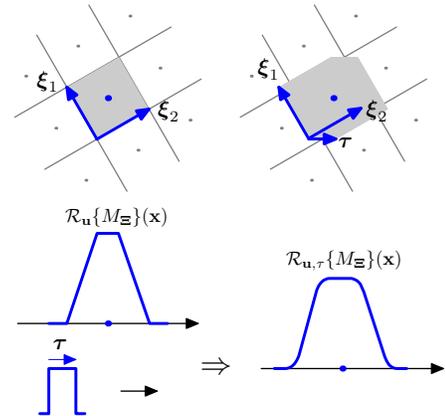}}
\caption{Projection of box spline with detector blur}\label{fig:boxSplineConvEff}
\end{figure}

\section{Forward Projection with Box Splines in Fan-Beam}\label{sec:fanbeamXray}
In fan-beam geometry, although the concept of X-ray and Radon transform of box splines still holds, the divergent nature of rays does not allow a direct application of (\ref{eq:projectionBoxSplineParallel}). The rays, in fan-beam geometry, are not perpendicular to the detector plane, however we can construct an orthogonal coordinate system for each ray by a rotation of detector plane and evaluating the projection at the corresponding coordinate system of individual rays. This transformation allows us to extend  the theory in parallel geometry to fan-beam setting. Unlike the parallel setting that projection results in a single box spline with a fixed direction set, $\mathbf{Z}$, for each viewing angle, we introduce a box spline for each viewing ray (i.e., projection point in the sinogram domain). 


\begin{theorem}\label{threorem:fan_parall}
The fan-beam projection of a box spline onto sinogram domain is a box spline whose direction set depends on and varies with the sinogram coordinate.
\end{theorem}
\begin{proof}
The X-ray projection of a box spline in image domain is the line integral:
\begin{equation}\label{eq:boxintegral}
\begin{split}
\Map \{M_{\boldsymbol {\Xi} }\}({s}) &= \int_{l_s(t)}M_{\boldsymbol {\Xi}}({\bf x}){\rm d}l_s(t) \\
& = \int_{0}^{\infty}M_{\boldsymbol {\Xi}}(\mathbf p + 
t( \frac{ \mathbf{p - P_{u^{\perp}}}s }{||\mathbf{p - P_{u^{\perp}}}s ||_2} ) ) {\rm d}t
\end{split}
\end{equation}
where $l_s(t)$ is the parametric equation of line. This line can also be reparameterized as:
\begin{equation}\label{eq:reparam_line}
l_s(t) = t \mathbf v(s) + \mathbf R_{\mathbf{v(s)^\perp}}^T s',
\end{equation}
where $\mathbf v(s)$ is the direction of the ray with unit length, $\mathbf{R}_{\mathbf{v}(s)^{\perp}}^T$ is a matrix whose columns span the hyperplane, $\mathbf{v}(s)^{\perp}$, perpendicular to ${\mathbf{v}(s)}$ and ${s}' = \mathbf{R}_{\mathbf{v}(s)^{\perp}}\mathbf{p}$ is the orthogonal projection of $\mathbf p$ on $\mathbf{v}(s)^{\perp}$. In this new parameterized form, the
matrix $[\mathbf R_{\mathbf v(s)^{\perp}}^T, \mathbf v(s)]$ forms a orthonormal frame rotated from image domain where the rotation angle varies with the sinogram coordinate, $s$.
With the reparameterized line equation (\ref{eq:reparam_line}) , the integral (\ref{eq:boxintegral}) is reformulated as:
\begin{equation}\label{eq:newboxintegral}
\begin{split}
\Map \{M_{\boldsymbol {\Xi} }\}({s}) =
\int_{0}^{\infty}M_{\boldsymbol {\Xi}}(t \mathbf v(s) + \mathbf R_{\mathbf{v(s)^\perp}}^T s') {\rm d}t,
\end{split}
\end{equation}
which is the exact form of the X-ray projection of box spline in parallel geometry in \cite{entezari2010box} and \cite{entezari2012box}.
Therefore the right-hand side of (\ref{eq:newboxintegral}) can be derived as:
\begin{equation}\label{eq:newboxintegralparallel}
\begin{split}
\int_{0}^{\infty}M_{\boldsymbol {\Xi}}(t \mathbf v(s) + \mathbf R_{\mathbf{v(s)^\perp}}^T s') dt &= \mathcal{R}_{\mathbf v(s)}\{M_{\boldsymbol {\Xi}}\}(s') \\
&=M_{\mathbf{R}_{ \mathbf{v}(s)^{\perp} } \boldsymbol{\Xi} }
(s')\\
&= M_{\boldsymbol Z (s) }({s'}). 
\end{split}
\end{equation}
The equations (\ref{eq:newboxintegral}) and (\ref{eq:newboxintegralparallel}) indicate
\begin{equation*}\label{eq:fanbeamBoxSplineProj}
\Map \{M_{\boldsymbol {\Xi} }\}({s}) = M_{\boldsymbol Z(s) }({s'}).
\end{equation*}
Because $M_{\boldsymbol Z(s) }$ is a box spline associated with direction set $\mathbf{Z}(s)$ varying with sinogram coordinate, the Theorem \ref{threorem:fan_parall} is proved. The geometric explanation of this proof is shown in Fig. \ref{fig:pointWiseBoxspline}.

\begin{figure}[ht]
\centering
{\includegraphics[width=0.9\columnwidth]{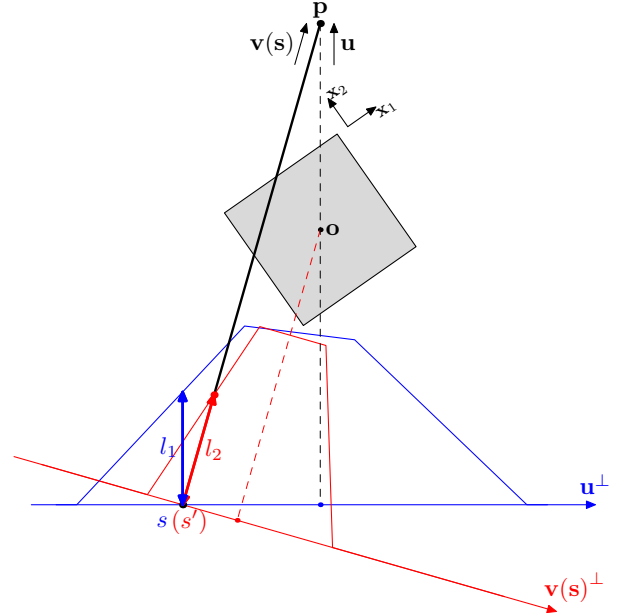}}
\caption{Fan-beam projection of box spline. The direction set of the projected box spline depends on coordinate $s$.}\label{fig:pointWiseBoxspline}
\end{figure}
\end{proof}
As a corollary we have:
\begin{corollary}
The fan-beam projection of the pixel-basis is a 2-direction box spline whose direction set varies with sinogram coordinates.
\end{corollary}
When the directions in $\boldsymbol \Xi = [\boldsymbol \xi_1, \boldsymbol \xi_2]$ are orthogonal, 
$M_{\boldsymbol \Xi}$ describes the pixel-basis. In the sinogram domain, direct evaluation
of box spline is possible by using rectangular box function corresponding to segments for each non-zero direction $\zeta(s) \in \boldsymbol Z(s)$:
\begin{equation*}
\begin{split}
M_{ \zeta(s)}(s)
&= \frac{1}{| \zeta(s)|}{\rm box}(\frac{s}{|\zeta(s)|}) \\
&= \frac{u(s) - u(s - \zeta(s))}{|\zeta(s)|}\\
&= \Delta_{\zeta(s)}u(s),
\end{split}
\end{equation*}
where $u(s)$ is the step function, and $\Delta_{\zeta(s)}$ denotes forward-differencing with a step size $\zeta(s)$. When $\zeta(s) = 0$, the elementary box spline reduces to a delta function, $M_0(s) = \delta(s)$, that gets eliminated in the convolutions in $\boldsymbol Z(s)$. According to the convolutional property of box splines, the fan beam X-ray projection of a pixel-basis can be expanded as:
\begin{equation}\label{eq:boxsplineProjectFanbeam}
\begin{split}
\Map \{M_{ \boldsymbol{\Xi} } \} (s) &= M_{\boldsymbol Z(s)}(s') =  M_{ \mathbf{R}_{\mathbf{v}(s) ^{\perp} }
\boldsymbol{\Xi} }(s') \\
&=M_{ [\zeta_1(s), \zeta_2(s)] }(s')\\
&=(\Delta_{\zeta_1(s)}u * \Delta_{\zeta_2(s)}u)(\mathbf{R}_
{\mathbf{v}(s)^{\perp} }\mathbf{p})\\
&=\Delta_{\zeta_1(s)} \Delta_{\zeta_2(s)}(\mathbf{R}_
{\mathbf{v}(s)^{\perp} }\mathbf{p})_+.
\end{split}
\end{equation}
Here the projected basis is a \emph{2-direction box spline} whose direction set $\boldsymbol Z(s)$ varies with the sinogram coordinate $s$. This box spline can be evaluated by forward differencing applied to the one-sided power function (i.e., ramp function here) defined as: $y_+ = \max(y, 0)$.

\subsection{Detector Blur}
For detector blur in \emph{fan-beam} geometry, because of the non-parallel structure of 
rays, the convolution of elementary box spline in image domain does not hold. Nevertheless, we derive an effective blur ${\tau}'$, which intersects an area with
pixel very close to the area cut by the detector source-bin triangle (solid blue) as illustrated in Fig. \ref{fig:pointWiseBoxsplineBinWidth}.
To derive the ${\tau}'$, we use the fact that ${\tau}'$ is a \emph{perspective projection} of ${\tau}$ on the plane $\mathbf{v}(s)^{\perp}$. Therefore, let $\mathcal{B}_{\mathbf{v}(s)}^T = [\mathbf{R}_{\mathbf v(s)^\perp}^T, \mathbf v(s)]$ denote the matrix whose columns are the basis vectors
of the coordinate system built by $\mathbf{v}(s)^{\perp}$ and $\mathbf{v}(s)$, $\mathcal{B}_{P}$ is the perspective projection matrix. By utilizing the homogeneous 
transformation, ${\tau}'$ is computed as:
\begin{equation}\label{eq:newDelta}
{\tau}' = 
\frac{(\mathcal{B}_{P}\mathcal{B}_{\mathbf{v}(s)}\ProjT{(s+\frac{\tau}{2} }))_{x_1}}
{(\mathcal{B}_{P}\mathcal{B}_{\mathbf{v}(s)}\ProjT{(s+\frac{\tau}{2} }))_{x_2}} -
\frac{(\mathcal{B}_{P}\mathcal{B}_{\mathbf{v}(s)}\ProjT{(s-\frac{\tau}{2} }))_{x_1}}
{(\mathcal{B}_{P}\mathcal{B}_{\mathbf{v}(s)}\ProjT{(s-\frac{\tau}{2} }))_{x_2}} .
\end{equation}
The division involved in transform (\ref{eq:newDelta}) is
\emph{perspective division} that is requisite after perspective transform, and $\tau'$ is achieved by differing the two transformed coordinates. At last,
\begin{figure}[ht]
\centering
{\includegraphics[width=0.7\columnwidth]{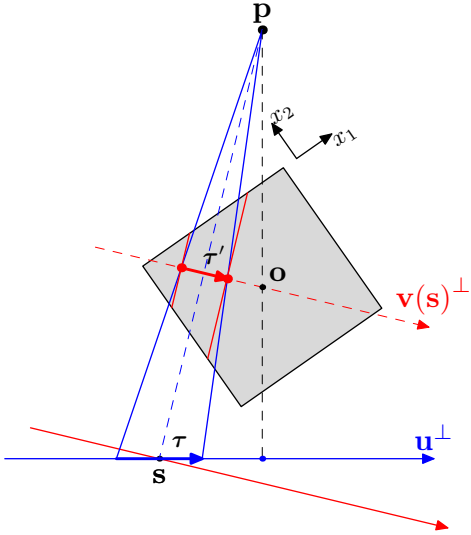}}
\caption{Projection of box spline with bin blur}\label{fig:pointWiseBoxsplineBinWidth}
\end{figure}
the combination of the effective blur (\ref{eq:newDelta}), equations
(\ref{eq:boxspineProjectionBinBlurParallel}) and 
 (\ref{eq:boxsplineProjectFanbeam}) from Theorem \ref{theory:detectorBlurParallel} and Theorem \ref{threorem:fan_parall} 
 yields the closed-form formula for X-ray transform of 
box spline in fan-beam with detector blur:
\begin{equation}\label{eq:pointwiseBoxSplineBlur}
\begin{split}
&\mathcal{P}_{\mathbf{u}, {\tau}} \{M_{ \boldsymbol{\Xi} } \} ({s}) 
= M_{ \mathbf{R}_{\mathbf{v}({s}) ^{\perp} } \boldsymbol{\Xi} }({s}') \\
&= (\Delta_{\zeta_1({s})}u * 
\Delta_{{\zeta}_2({s})}u * \Delta_{{\tau}'}u)
(\mathbf{R}_{\mathbf{v}({s})^{\perp} }\mathbf{p})\\
&=\frac{\Delta_{\zeta_1({s})}  \Delta_{\zeta_2({s}) }  \Delta_{{\tau}'}
(\mathbf{R}_{\mathbf{v}({s})^{\perp} }\mathbf{p})_+^2}{2!}.
\end{split}
\end{equation}
Equation (\ref{eq:pointwiseBoxSplineBlur}) is the analytical closed-form approximation
of X-ray transform of a box spline in fan-beam geometry.
In a common fan-beam setup, the size of pixels and detector bin width are usually very
small and projector-detector distance is much larger so that the approximation made by the effective blur is relatively highly accurate and this will be demonstrated in experiment section.

\section{Experiments and results}

The goal in the proposed convolutional non-separable footprint (CNSF) formalism is to provide an efficient and accurate method to model the forward and back-projection in fan-beam X-ray CT problem, as these operations are the fundamental computations in almost all the iterative reconstruction algorithms. To examine the accuracy and efficiency of our CNSF framework, we compare it with the
state-of-art algorithms designed for efficient computation of forward and back-projection, namely: the Separable Footprints (SF)~\cite{long20103d} and the Look-up Table-based Ray Integration (LTRI)~\cite{ha2018look}. For a \emph{reference projector} (Ref), as the function  (\ref{eq:boxsplineProjectFanbeam}) is the exact formula for X-ray transform
of pixel-basis in fan-beam without detector blur, we use symbolic function to evaluate the integral in an interval specified by detector bin width with relative error tolerance to $0$ and absolute error tolerance to $10^{-12}$, where the computation is extremely expensive. 

\subsection{Forward Projection for Single Pixel}
In order to demonstrate accuracy of approximation made by our method, we first show the ``microscopic'' view of 
the projections of a single pixel achieved by different methods.
We simulate a 2-D fan-beam flat detector X-ray CT system with a single pixel centered at rotation center with size $(1 \text{ mm} \times 1\text{ mm})$. The detector bin width $\tau$ is $0.5$ mm, the projector to rotation center distance $D_{po} = 3$ mm, and detector to rotation center distance $D_{so} = 3$ mm.  
In order to visualize the approximations made by different methods as continuous functions, we over-sample the projections
by setting the interval of detectors $\Delta_s$ to $0.01$ mm and the number of detectors $N_s$ to  $601$. 

\begin{figure*}[ht!]
     \begin{minipage}[l]{0.5\columnwidth}
         \centering
         \includegraphics[width=4.5cm]
         {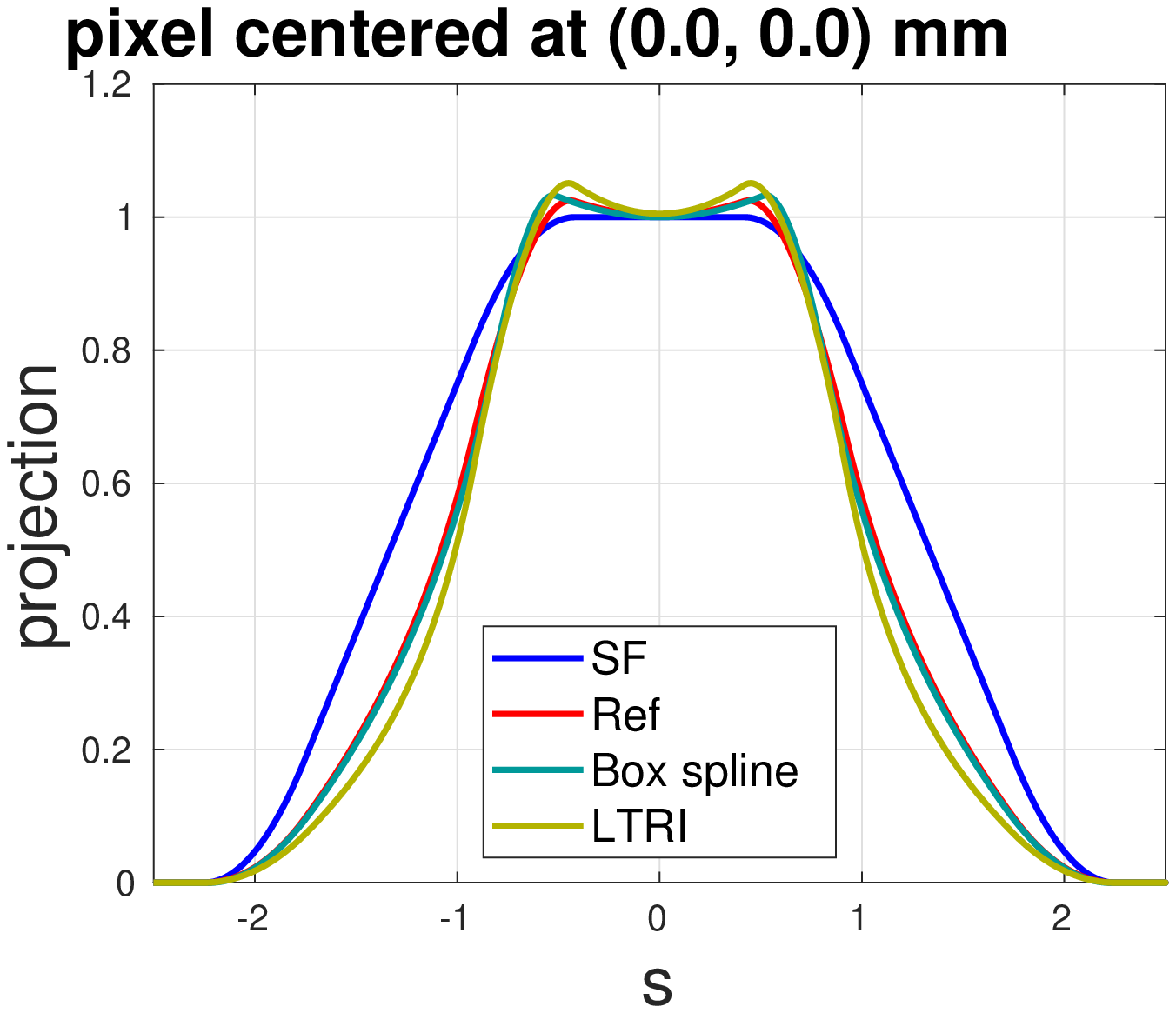}
         \subcaption{}\label{fig:AAA}
     \end{minipage}
          \begin{minipage}[l]{0.5\columnwidth}
         \centering
         \includegraphics[width=4.5cm]
         {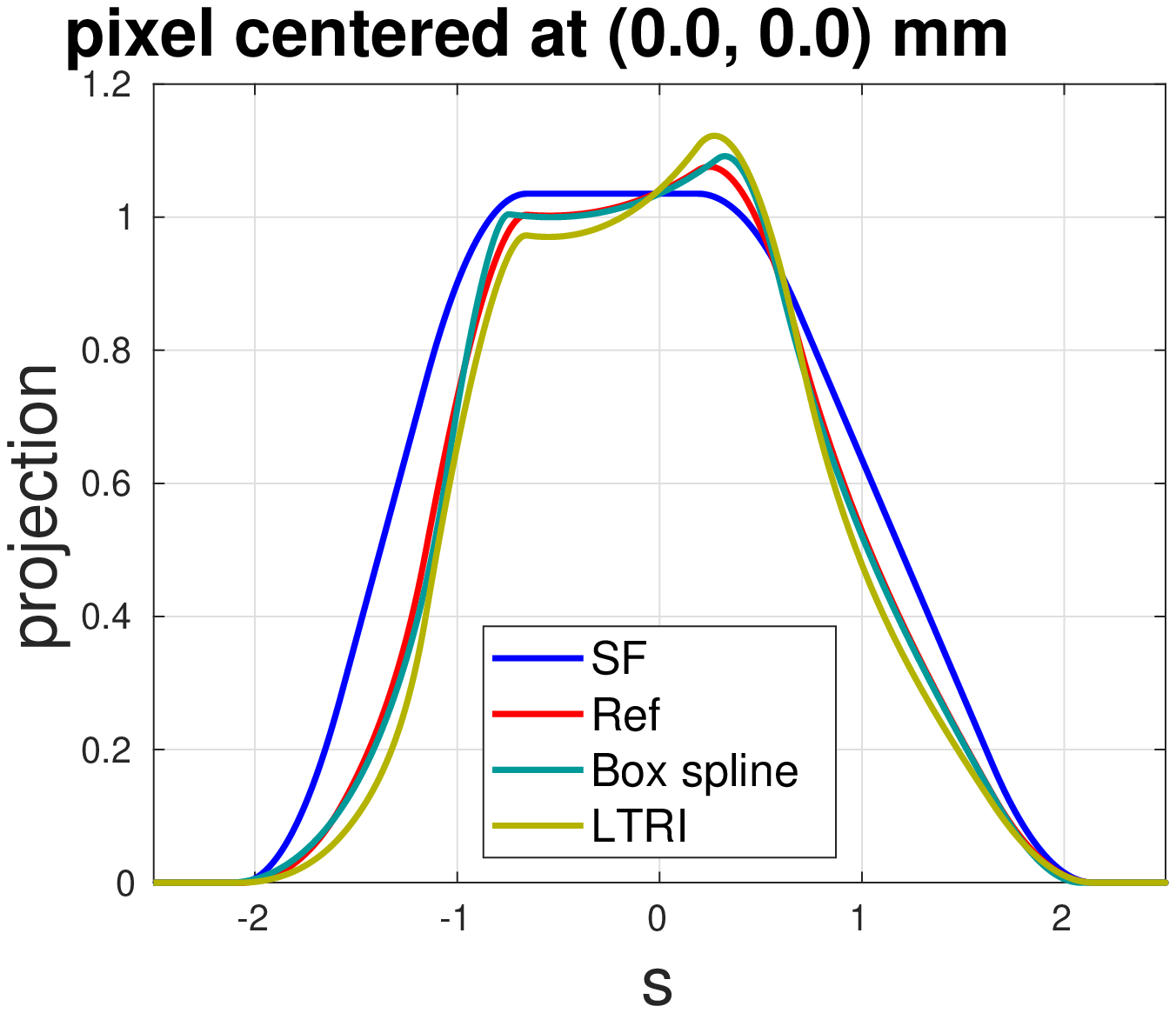}
         \subcaption{}\label{fig:BBB}
     \end{minipage}
          \begin{minipage}[l]{0.5\columnwidth}
         \centering
         \includegraphics[width=4.5cm]
         {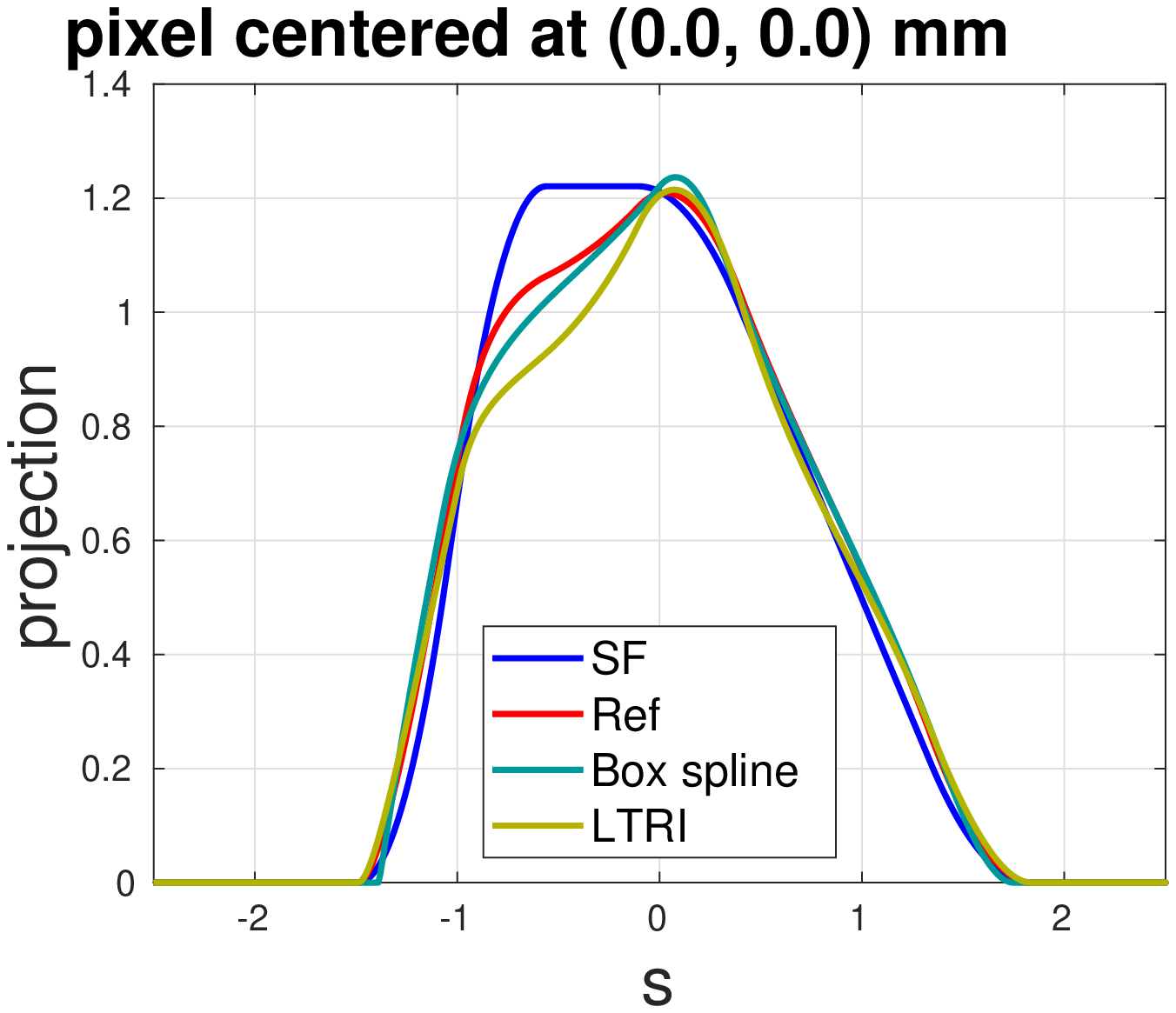}
         \subcaption{}\label{fig:CCC}
     \end{minipage}
          \begin{minipage}[l]{0.5\columnwidth}
         \centering
         \includegraphics[width=4.5cm]
         {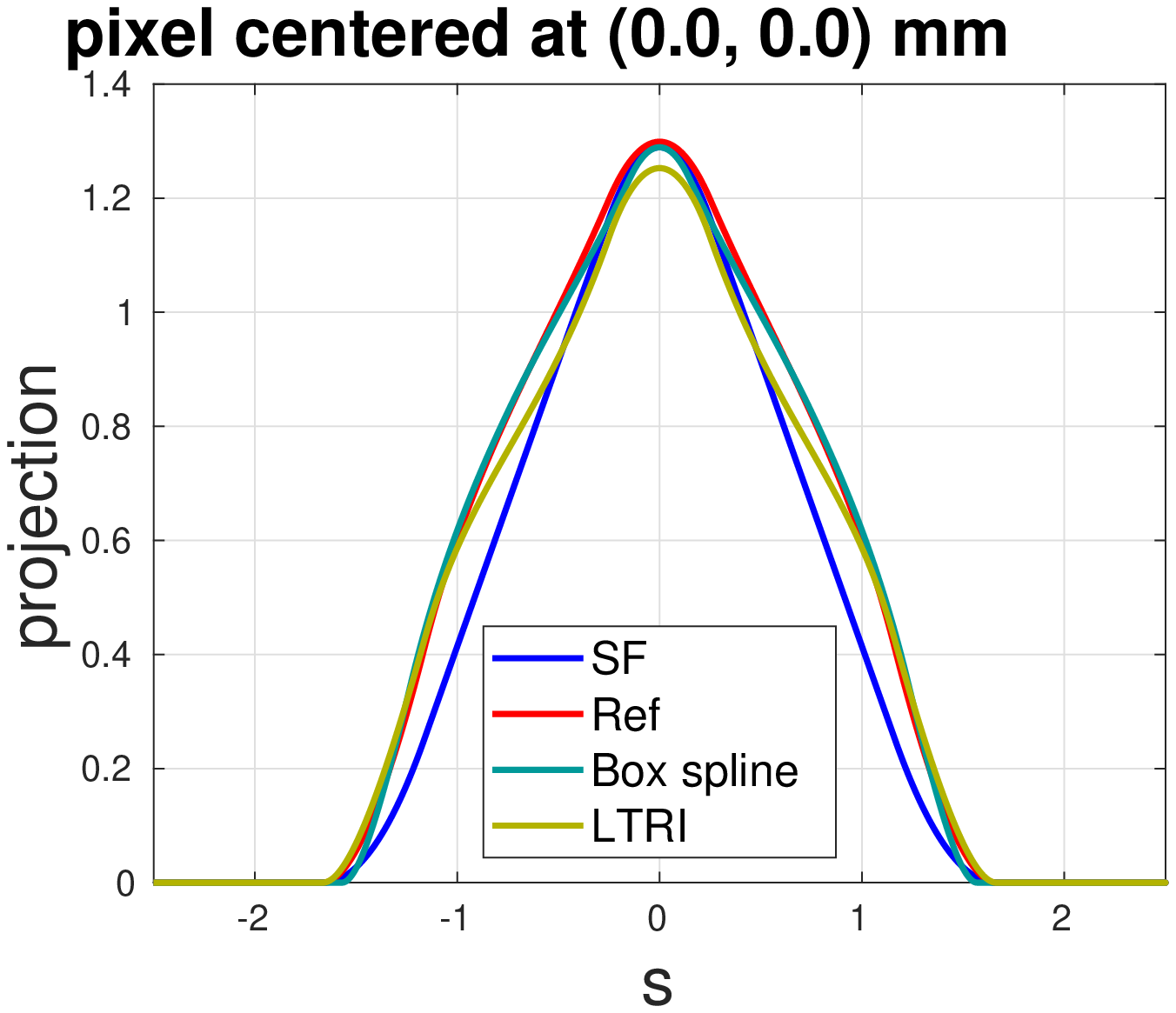}
         \subcaption{}\label{fig:DDD}
     \end{minipage}
   \caption{From (a)-(d), the projection angles are $0^{\circ}$, 
   $15^{\circ}$, $35^{\circ}$ and $45^{\circ}$ respectively. }
  	\label{fig:experiment_1PixelProjBlur}
\end{figure*}

We use the MIRT toolbox \cite{irtFessler} and LTRI source code  \cite{ltriHa} to generate the comparisons and we implement the CNSF
projector in CUDA. Fig. \ref{fig:experiment_1PixelProjBlur} shows the projection of pixel-basis for all four projectors with detector blur. Our proposed method in all these angles is capable to closely approximate the nonlinear functions. SF can only use trapezoid-shape functions to approximate the functions, but obviously in most situations, the analytical projection function has irregular shape. Although LTRI can also fit the curve, the accuracies are much less than our method in most locations in
sinogram domain.

The geometric setting in above experiment is to illustrate the performance of the approximations made by the compared algorithms in extremely detailed level, while
in order to show the comprehensive performance of the projection in more practical setting, we compare the 
\emph{maximum errors }of forward projectors by defining:
\begin{equation}
e(\mathbf{u}) = \max_{s \in \mathbb{R} }|F(s, \mathbf{u}) - F_{ref}(s, \mathbf{u})|,
\end{equation}
where $F$ is the projection function approximated by any of SF,  LTRI, and CNSF projectors, and $F_{ref}$ is generated by {reference} projector.
The pixel is centered at $(0,0)$ mm and $(100.5, 50.5)$ mm respectively with size $(1 \text{ mm}\times 1\text{ mm}) $, 
$\tau = 0.5$ mm, $\Delta_s = 0.5$ mm, $D_{po} = 200$ mm, and $D_{so} = 200$ mm. Since the pixel centered at $(0, 0)$ mm is symmetric in all four
quadrants, we only evenly select $90$ angle over $90^{\circ}$ that is shown in 
Fig. \ref{subfig:maxErrorA}, while in Fig. \ref{subfig:maxErrorB},  
$360$ angles over $360^{\circ}$ are uniformly selected.
\begin{figure}[ht!]
    \begin{subfigure}{.47\linewidth}
        \includegraphics[scale=0.18]{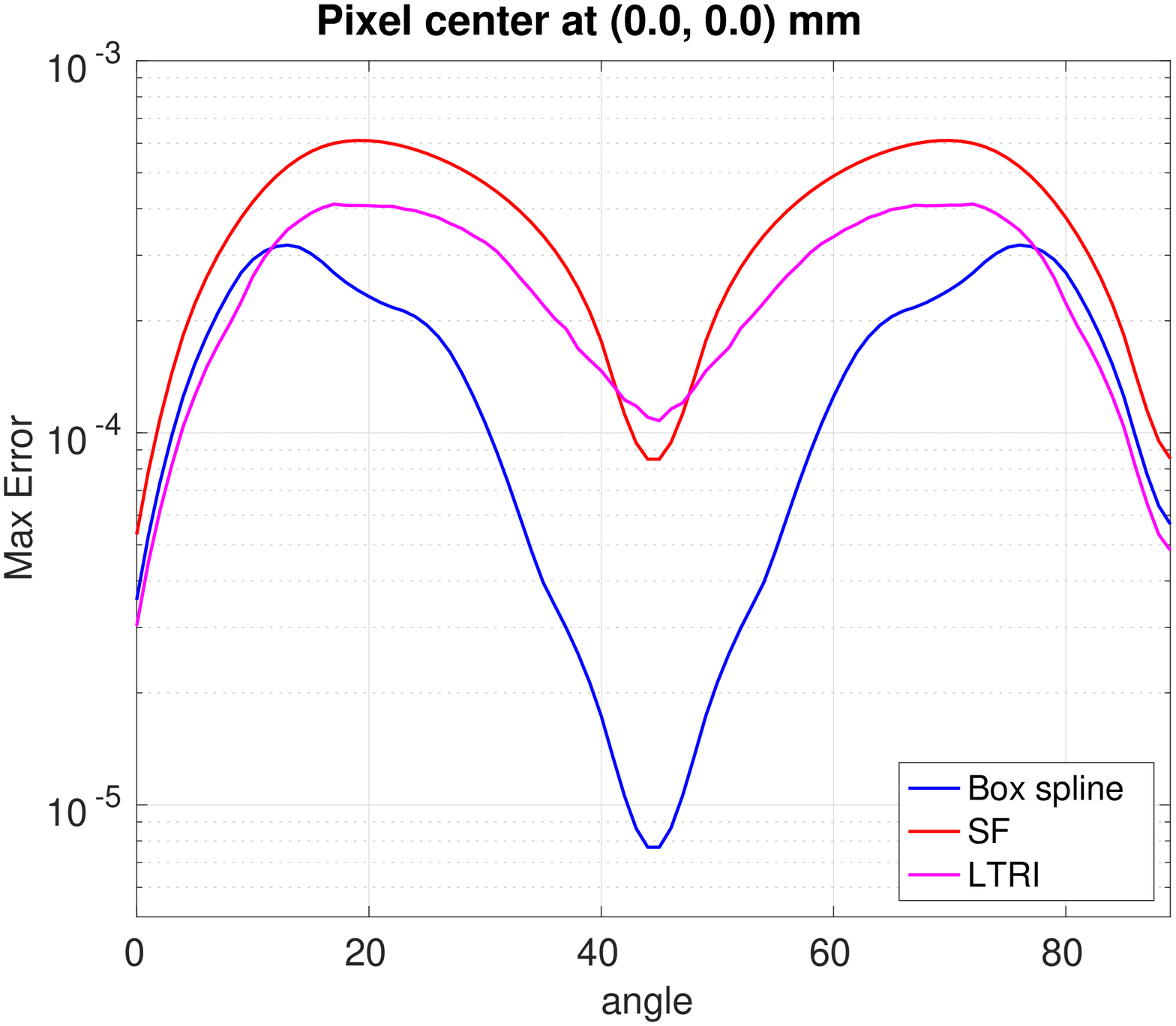}
        \caption{}\label{subfig:maxErrorA}
    \end{subfigure}
    \hskip1em
    \begin{subfigure}{.47\linewidth}
        \includegraphics[scale=0.18]{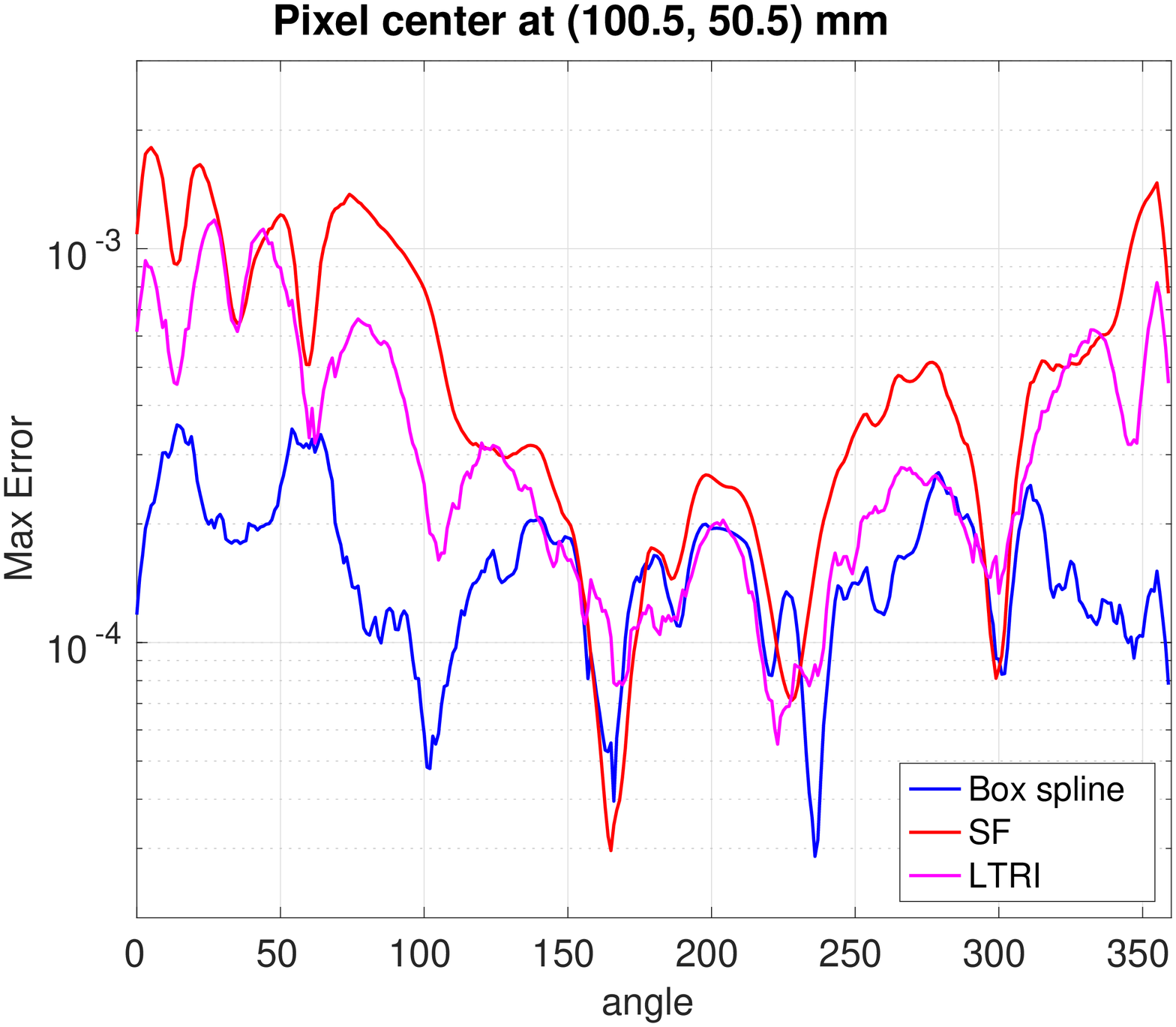}
        \caption{}\label{subfig:maxErrorB}
    \end{subfigure}
    \caption{Maximum errors comparison among CNSF, SF and LTRI projectors for a pixel centered at different locations.}\label{fig:experiment_BoxplineProjBlur}
\end{figure}

When the pixel is located at origin, the errors presented in LTRI are slightly smaller than CNSF at the angles ranged from $0^{\circ}$ to $10^{\circ}$ and $78^{\circ}$ to $90^{\circ}$, whereas our method performs much better at other angles. If pixel offsets a large distance from origin, the accuracy will degrade as in \ref{subfig:maxErrorB}, but our method suffers the least impact from the asymmetric location 
of pixel. Overall, the proposed method outperforms other methods in most angles of forward projections in respect of accuracy.

\subsection{Radon Transform of Image}
We compare the accuracies of forward projection for complete 2-D images with a size of $(128 \text{ mm} \times 128 \text{ mm})$ and a size of $(256 \text{ mm}\times 256\text{ mm})$.  The
numbers of detectors $N_s$ are $409$ and $815$ respectively, and they are spaced by $\Delta_s = 1$ mm with bin width $\tau = 0.5$ mm.

Fig. \ref{fig:RadonTransformCompare} shows the absolute errors of forward projections from 360 angles of Shepp-Logan and Brain phantom, and the errors are scaled by $512$ and $1024$ respectively for visualization purposes. The projector-rotation center
and detector-rotation center distance pair, $(D_{po}, D_{so}$), is $(200,200)$ mm for Shepp-Logan and $(400, 400)$ mm for brain. The results substantiate the significant improvements in accuracy compared to LTRI and SF.
\begin{figure}[ht!]
\begin{tabular}{ccccr}
\includegraphics[width=0.21\linewidth]{./experiment_radon_groundTruth}&
\includegraphics[width=0.21\linewidth]{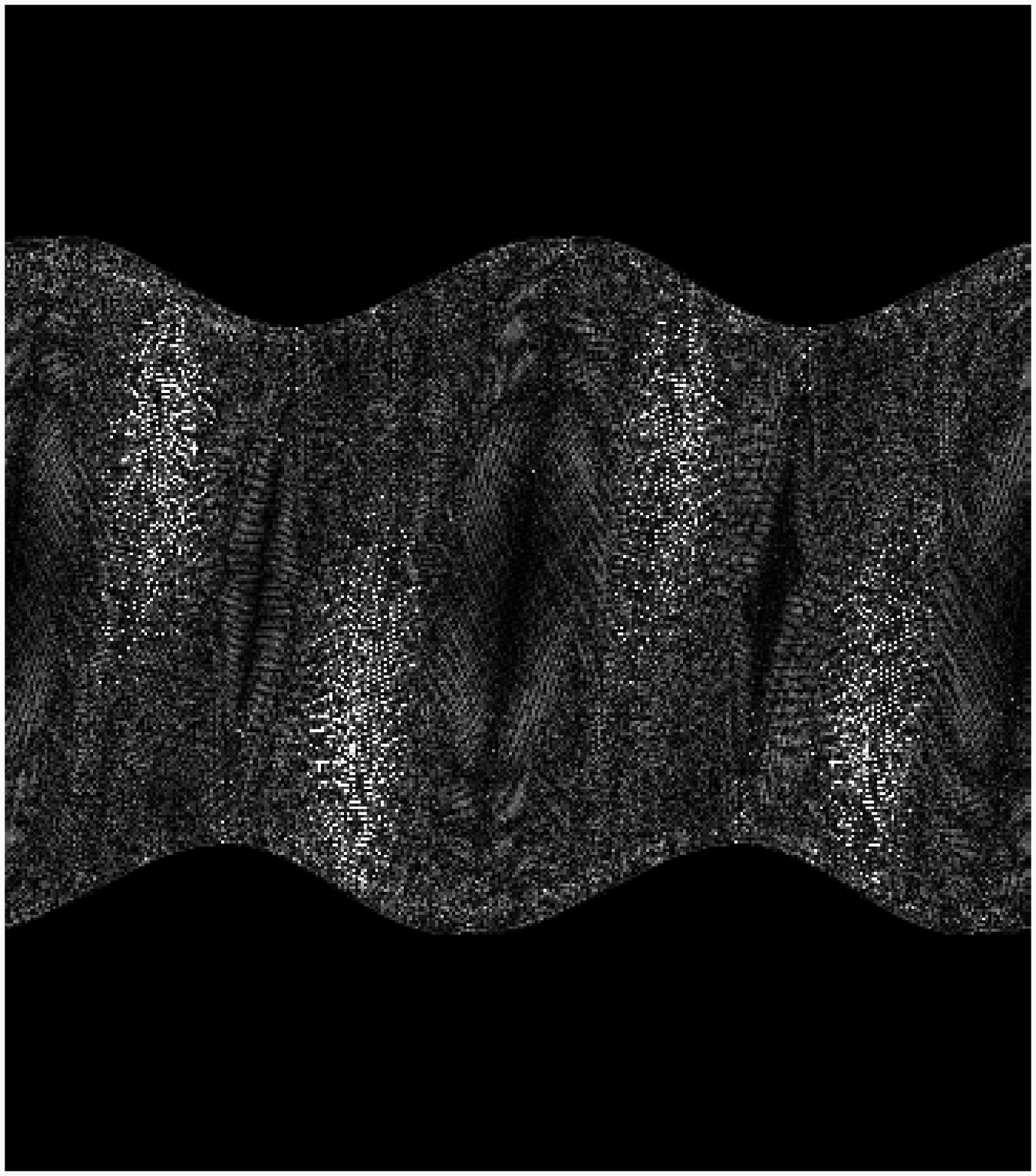}&
\includegraphics[width=0.21\linewidth]{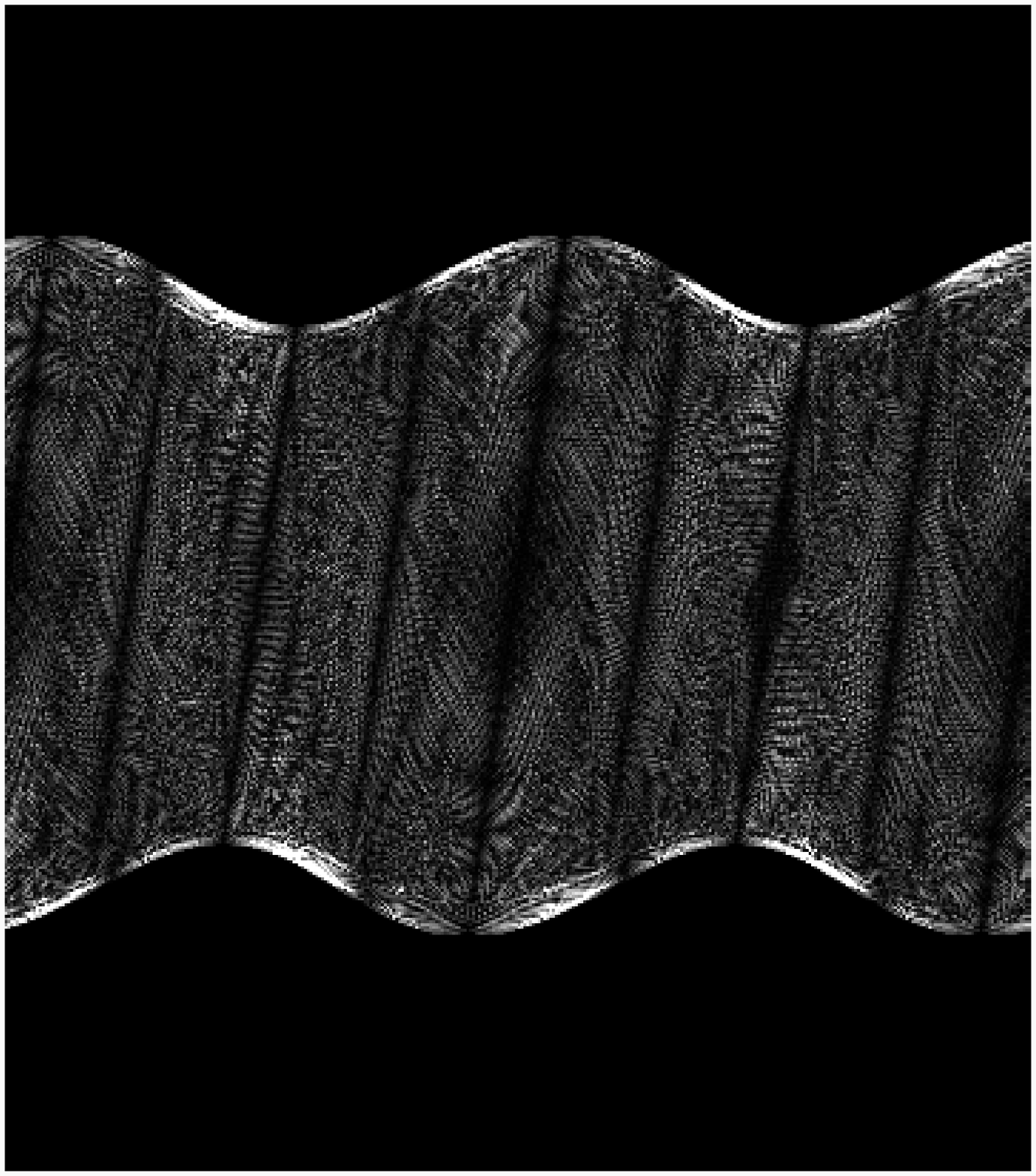}&
\includegraphics[width=0.21\linewidth]{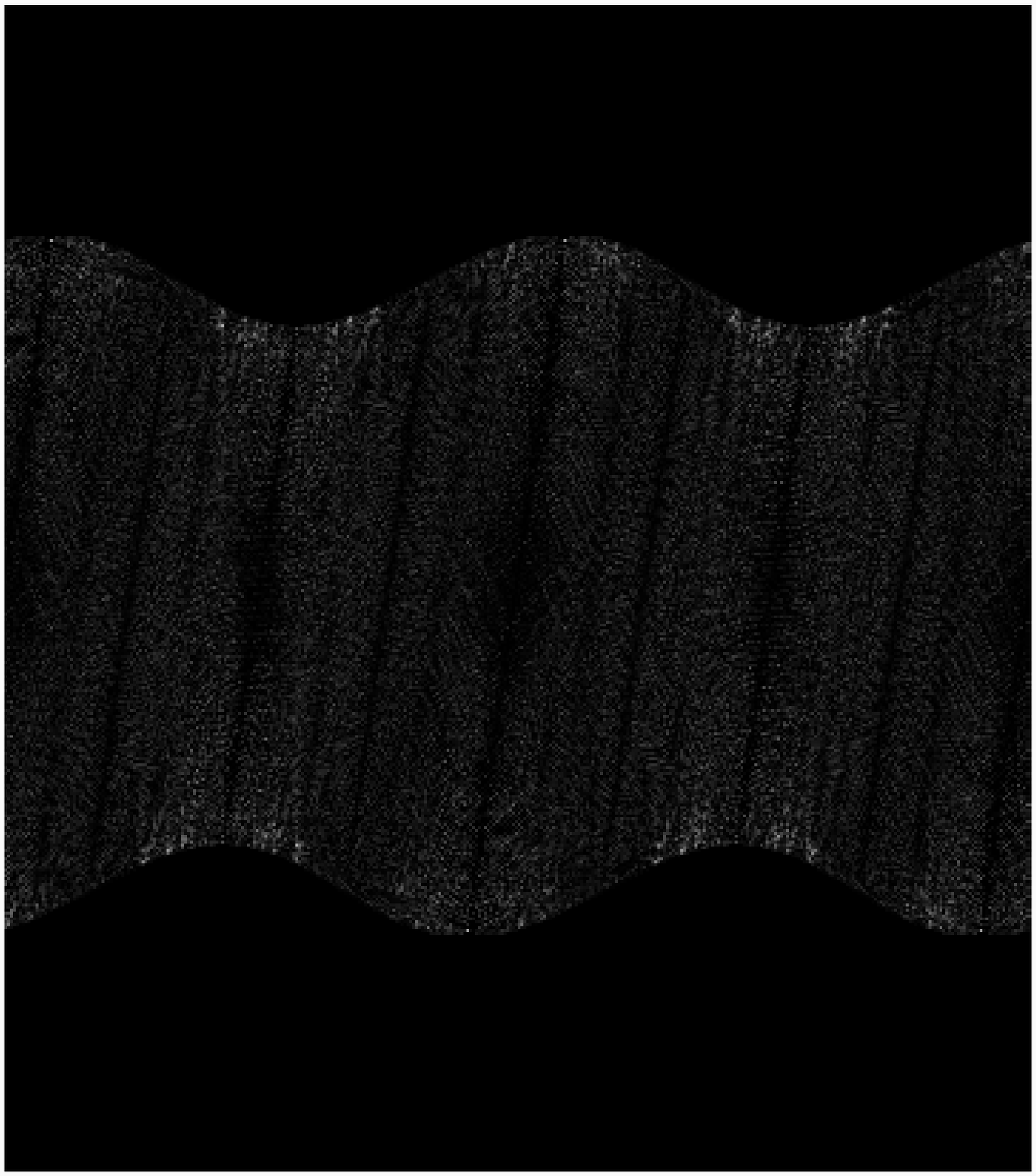}\\
\includegraphics[width=0.21\linewidth]{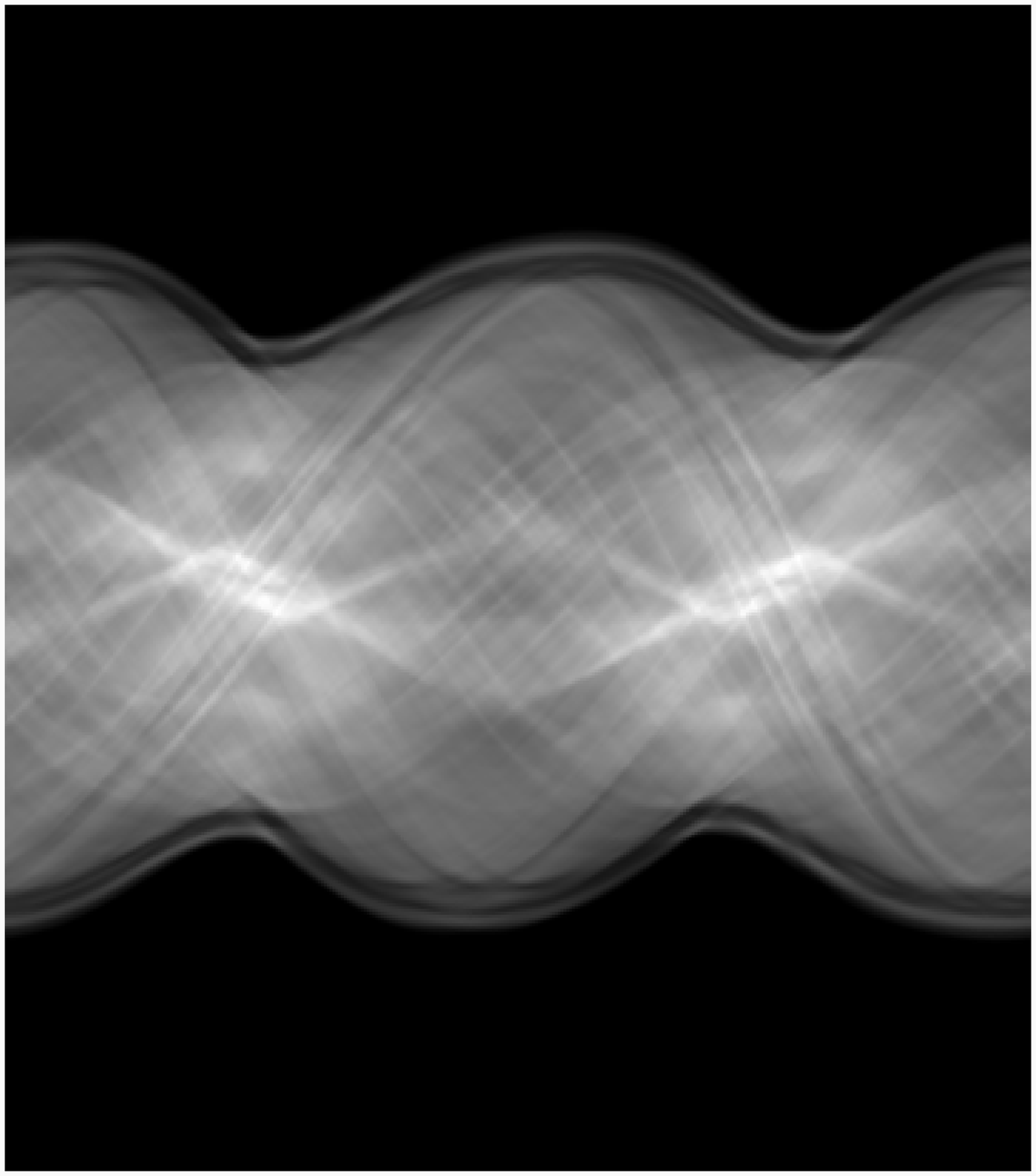}&
\includegraphics[width=0.21\linewidth]{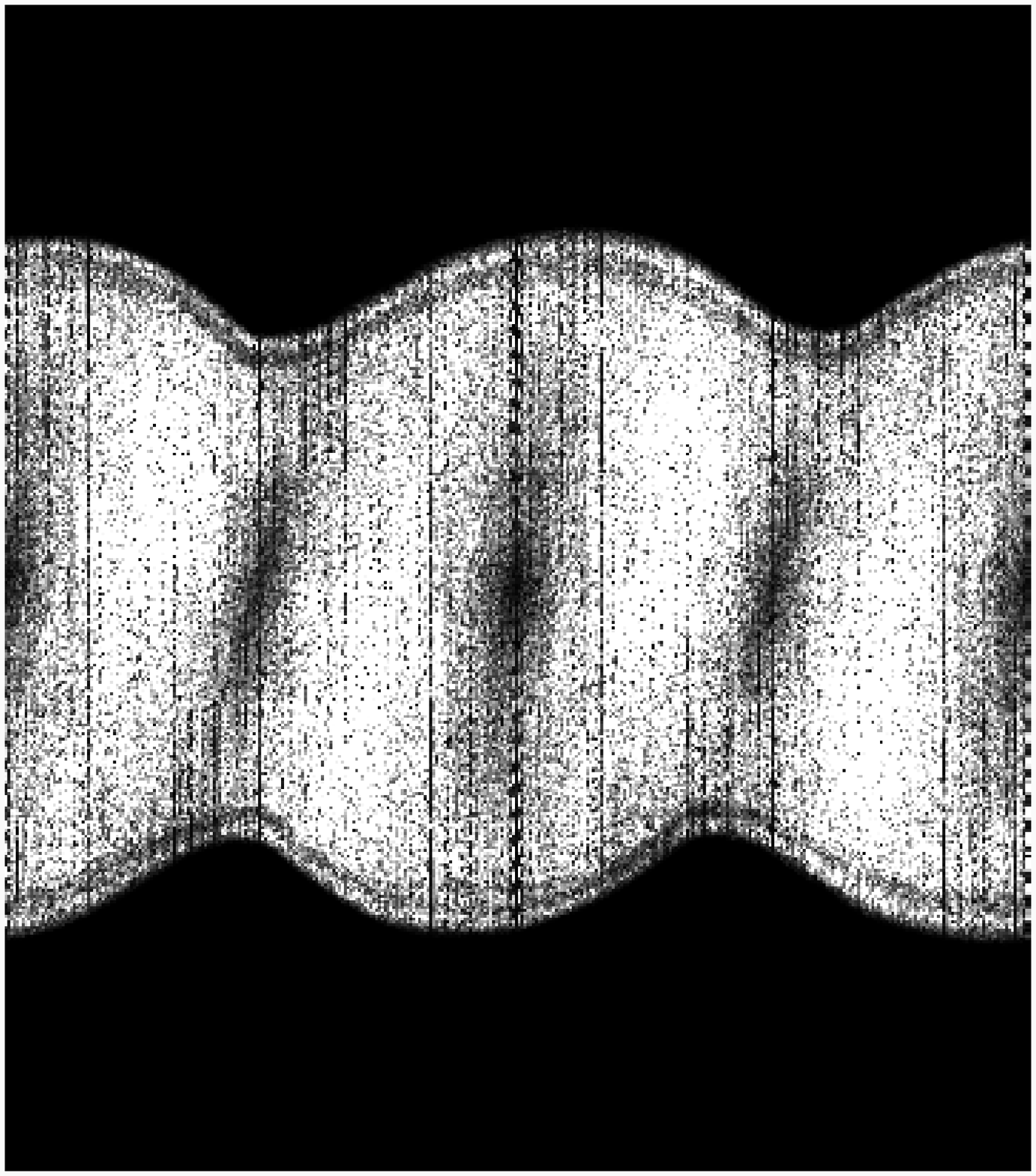}&
\includegraphics[width=0.21\linewidth]{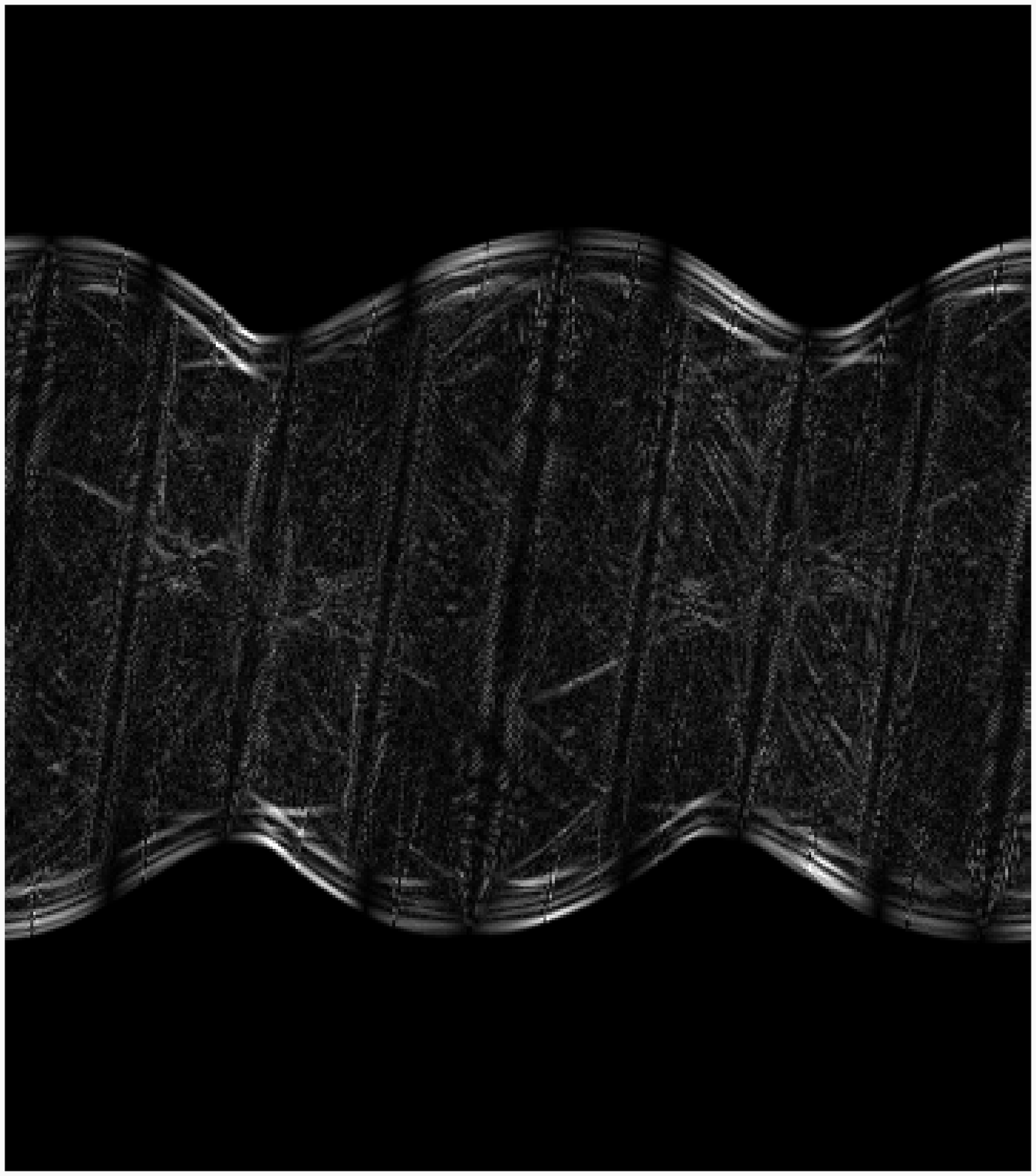}&
\includegraphics[width=0.21\linewidth]{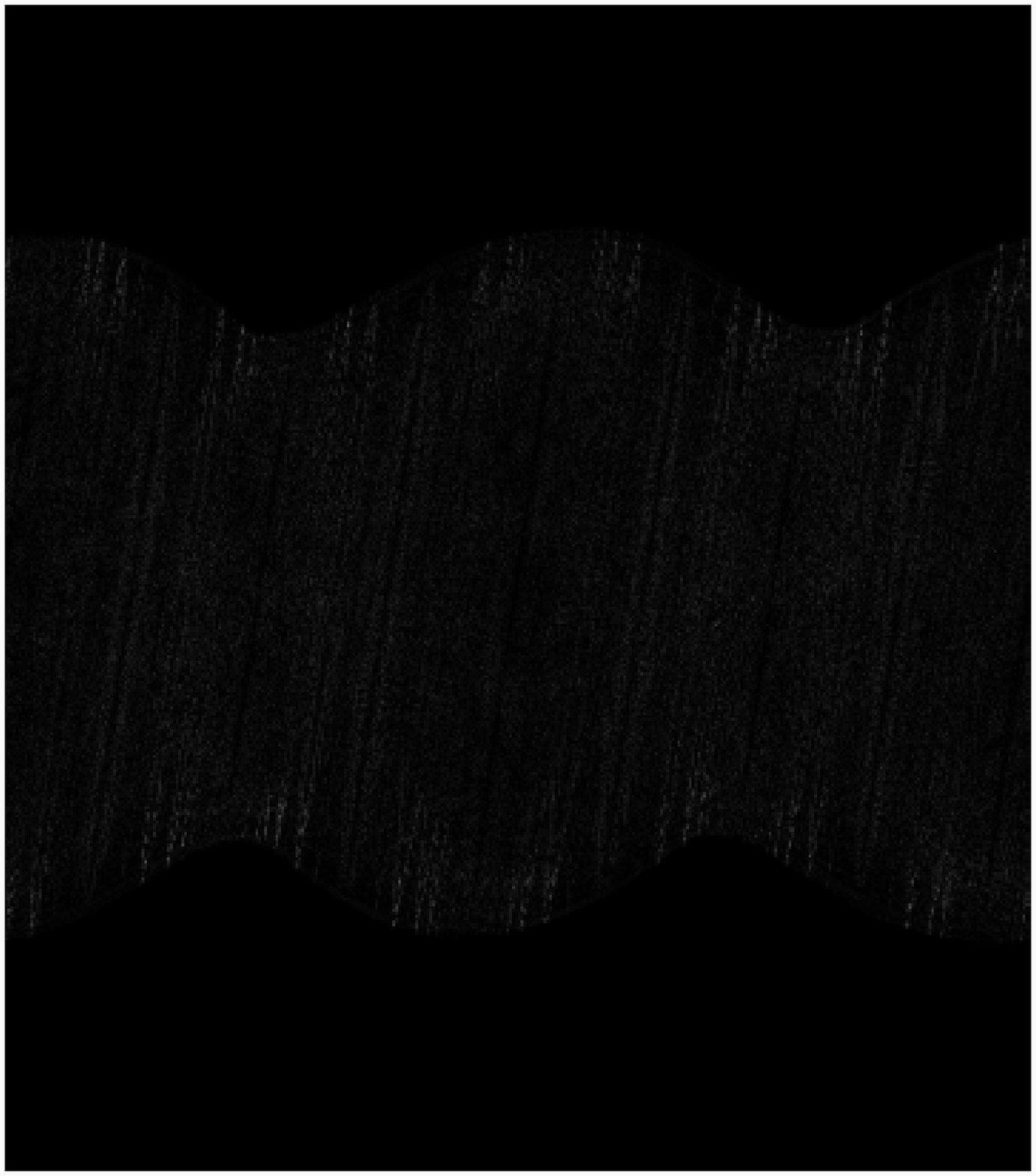}\\
  Ref&LTRI&SF&CNSF
\end{tabular}
\caption{Absolute error in the projection (sinogram) domain. First row shows the projection
of Shepp-Logan dataset with size of $128 \text{ mm} \times 128\text{ mm}$ while second row is brain dataset with
size of $256 \text{ mm}\times 256\text{ mm}$.} \label{fig:RadonTransformCompare}
\end{figure}
It is also worth noting the fact that when the projector and detector are located farther away each other, as the second row, the error becomes smaller, which also corroborates our statement in section \ref{sec:fanbeamXray}.

\subsection{Time performance of Forward and Back-projection}
The other important advantage of our proposed method is
its ability to achieve high performance with on-the-fly computations of the forward projection (FP) and back-projection (BP), eliminating the
needs to store the system matrix and/or a look-up table. This is accomplished
by efficient evaluation of the box spline on the right hand
side of (\ref{eq:pointwiseBoxSplineBlur}).
In this experiment, we simulate an flat detector fan-beam X-ray CT system with $360$ angles over $360^{\circ}$ and bin width $\tau = 1$ mm. The image resolutions used in the simulation are ($64 \text{ mm}\times 64\text{ mm} $), ($128 \text{ mm}\times 128 \text{ mm}$), ($256 \text{ mm}\times 256\text{ mm}$), ($512\text{ mm}\times 512\text{ mm}$), ($1024\text{ mm}\times 1024\text{ mm}$), and ($2048\text{ mm}\times 2048\text{ mm}$). The images are all-ones images (the intensity of each pixel in the image is 1) for equitable comparison. The detector bin is spaced by $\Delta_s = 1$ mm, and the numbers of the detector $N_s$ are $205$, $409$, $815$, $1627$, $3250$, and $6499$ corresponding to the different image resolutions. 
In order to adapt different field-of-views and image sizes, several projector-rotation center and detector-rotation center distances are selected as $(D_{po}, D_{so})$ = 
($100, 100$) mm, ($200, 200$) mm, ($400, 400$) mm,  ($800, 800$) mm, 
($1600, 1600$) mm, and ($3200, 3200$) mm.
Fig. \ref{fig:timeComparisonWithLTRI} shows our speedup factors over LTRI for the average projection of all angles. All experiments are performed on NVIDIA-GTX 1080 GPU with CUDA-10.0, Intel i7-6800K 6-core CPU with 3.40 GHz.
\begin{figure}[ht]
\centering
{\includegraphics[width=1\columnwidth]{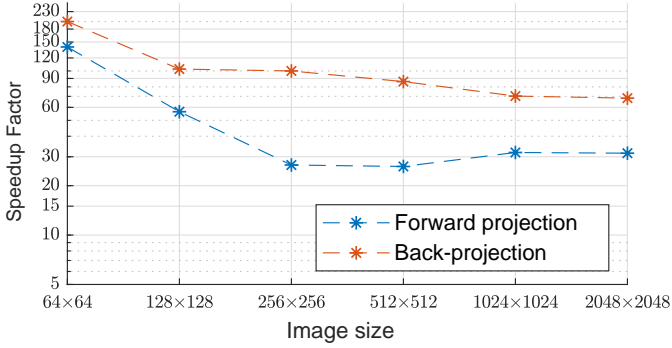}}
\caption{Run time comparisons with LTRI}
\label{fig:timeComparisonWithLTRI}
\end{figure}
Elimination of the necessity to access a look-up table leads
to high throughput in GPU implementations. Our method does not store any intermediate data for projections, unlike the precomputed table in LTRI. Therefore, there is not a lot of GPU
memory access, which is usually the bottleneck in GPU computing, in our implementation.
It is also noteworthy that the speedup in back-projection is always higher than forward projection in
all resolutions. The reason for this phenomenon is that in forward projection, each projection value
is weighted sum of several pixels, thus, in CUDA implementation, each kernel thread will
execute the atomic add instruction that degrades the efficiency of parallel computing.
This serialized behavior from atomic operation occurs much less in back-projection of our implementation leading to improved speedup.

Since there is no publicly available GPU version of SF method, we also implement the CPU version of proposed method with Intel Threading Building Blocks (TBB) library that can parallelize code and compare the run time with SF in CPU parallel computing.
\begin{figure}[ht]
\centering
{\includegraphics[width=1\columnwidth]{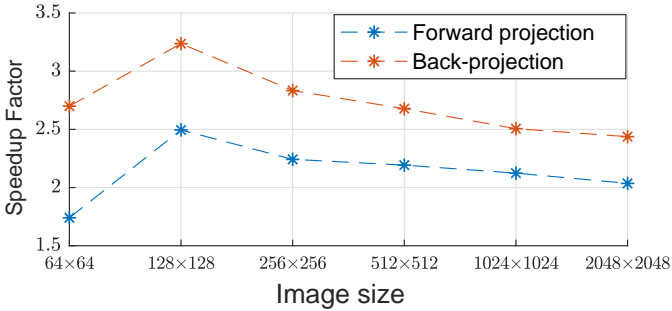}}
\caption{Run time comparisons with SF}
\label{fig:timeComparisonWithSF}
\end{figure}
Fig. \ref{fig:timeComparisonWithSF} shows the speedup of average projection of
all angles over SF. Since the computational resources on CPU are much less than
GPU, the speedups are not so significant as ones achieved by GPU. However, our method can still achieve as least 2 times speedups over SF.

\subsection{Reconstruction Performance}
To evaluate the improvements of accuracy in our method in the actual 
reconstruction problems, we perform experiments where the measurement was
generated by the \emph{reference projector} and reconstructed with different approximating projectors (i.e., SF, LTRI and CNSF). 
The optimizer used in the experiments is
adaptive steepest descent POCS regularized by total-variation (TV) called
ASD-POCS \cite{sidky2008image} and all the
hyper-parameters are adjusted to achieve the highest reconstruction quality for each model.

In this experiment, we use Shepp-Logan and the Brain phantom with resolution ($128 \text{ mm}\times 128 \text{ mm}$)  and ($256 \text{ mm}\times 256\text{ mm}$) respectively as the benchmark dataset. The simulated flat detector X-ray system is configured with $N_s = 409$, $D_{po} = 200$ mm, $D_{so} = 200$ mm and $N_s = 815$, $D_{po} = 400$ mm, $D_{so} = 400$ mm. The detectors are spaced by $\Delta_s = 1$ mm with bin width $\tau = 1$ mm and the projectors are uniformly spaced over $360^{\circ}$. Fig. \ref{fig:reconstructionSheppLogan} shows the reconstruction result of Shepp-Logan phantom. The (imperfect) reconstruction achieved by reference
projector in Fig. \ref{fig:reconstructionSheppLoganRef} illustrates the
practical reconstruction problem from limited-view projection. The approximating projector made by LTRI results in a less accurate reconstruction (the resolution of look-up table in LTRI is $10000 \times 180$). SF makes a slightly more accurate approximation and achieves a little higher quality over LTRI. Our method provides a reconstruction
that can achieve the same quality as the one provided by reference projector.

\begin{figure}[ht!]
    \centering
\hskip1em
    \begin{subfigure}{.37\linewidth}
        \includegraphics[scale=0.4]{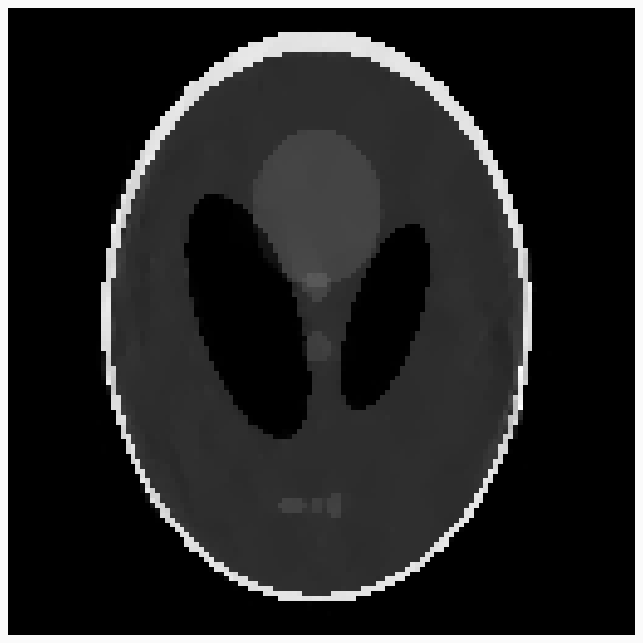}
        \caption{Ref, SNR = 26.39dB}\label{fig:reconstructionSheppLoganRef}
    \end{subfigure}
    \hskip3em
    \begin{subfigure}{.4\linewidth}
        \includegraphics[scale=0.4]{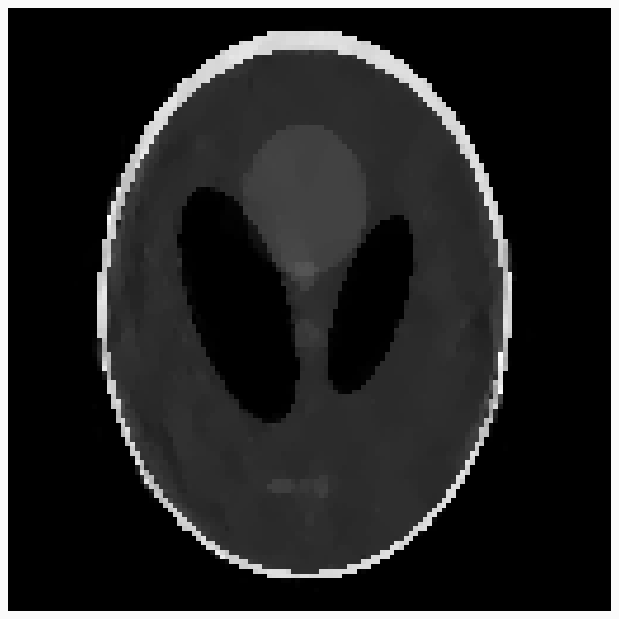}
        \caption{LTRI, SNR = 25.63dB}
    \end{subfigure}
	\vskip2em
\hskip1em
    \begin{subfigure}{.37\linewidth}
        \includegraphics[scale=0.4]{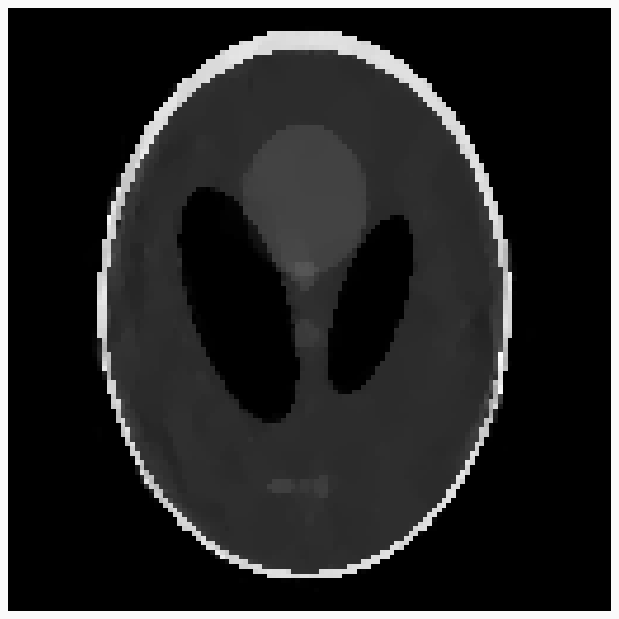}
        \caption{SF, SNR = 25.80dB}
    \end{subfigure}
	\hskip3em
	\begin{subfigure}{.4\linewidth}
     	\includegraphics[scale=0.4]{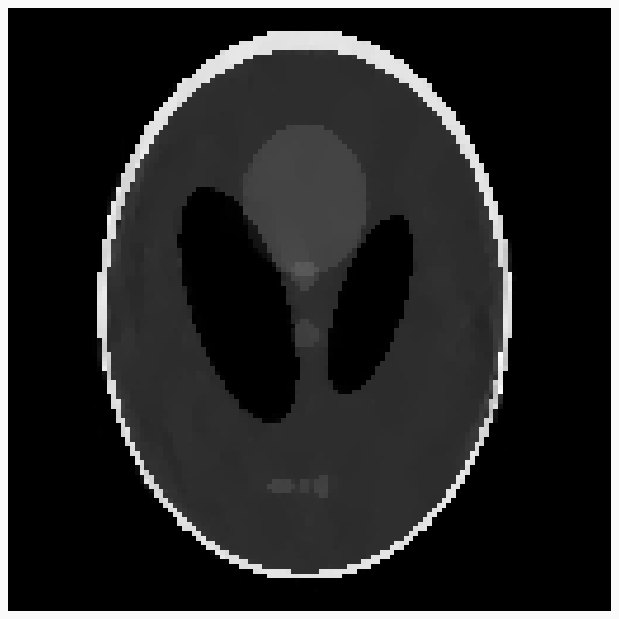}
        \caption{CNSF, SNR = 26.39dB}
    \end{subfigure}
    \caption{Reconstruction of Shepp-Logan phantom from $16$ uniformly spaced projections using ASD-POCS. }\label{fig:reconstructionSheppLogan}
\end{figure}

Fig. \ref{fig:reconstructionBrain} shows the reconstruction of the brain phantom. In order to evaluate the impact of the accuracy of the
forward model in image reconstruction, we visualize the differences of
all reconstructions from the reconstruction provided by 
reference model that is shown in \ref{fig:reconstructionBrainRef}. For 
visualization purpose, we scale the errors by appropriate factors shown
in captions. The SNRs for these results are (Ref) $19.38$dB, (LTRI) $18.96$dB, (SF) $19.07$dB, (CNSF) $19.38$dB respectively.
\begin{figure}[ht!]
    \centering
    \begin{subfigure}{.45\linewidth}
        \includegraphics[scale=0.32]{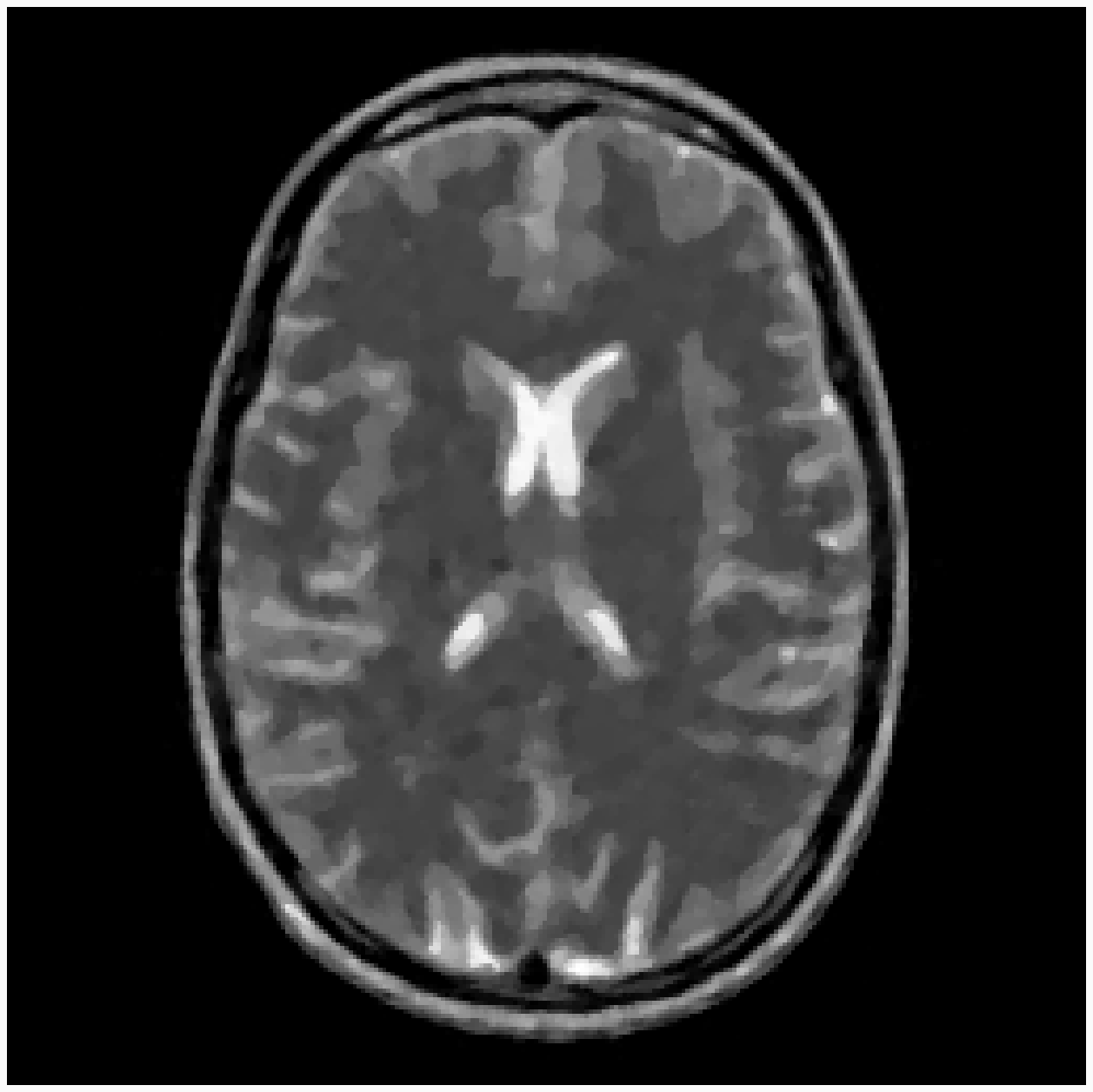}
        \caption{Ref}\label{fig:reconstructionBrainRef}
    \end{subfigure}
    \hskip2em
    \begin{subfigure}{.45\linewidth}
        \includegraphics[scale=0.33]{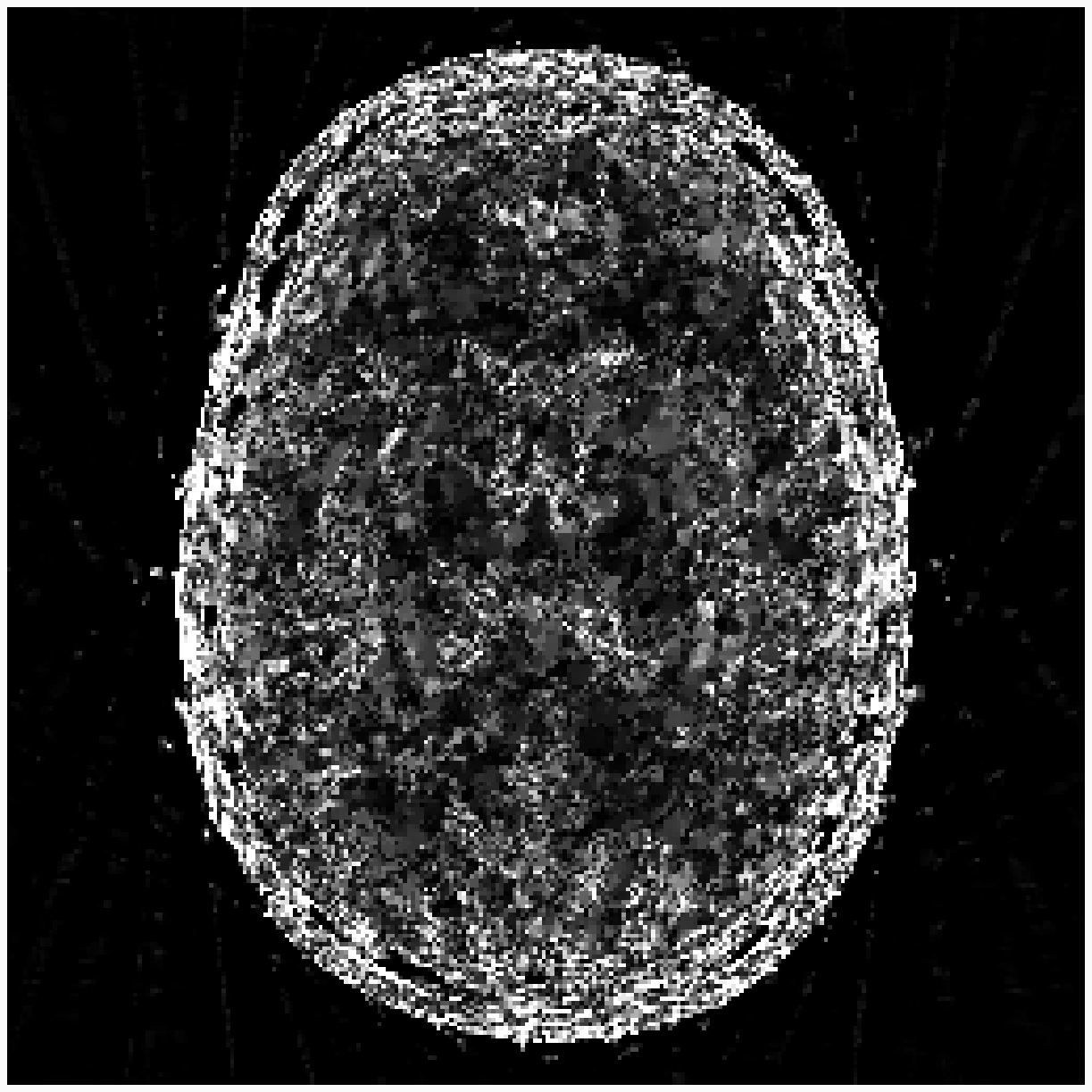}
        \caption{(LTRI - Ref)$\times 10^2$}
    \end{subfigure}
	\vskip2em
    \begin{subfigure}{.45\linewidth}
        \includegraphics[scale=0.33]{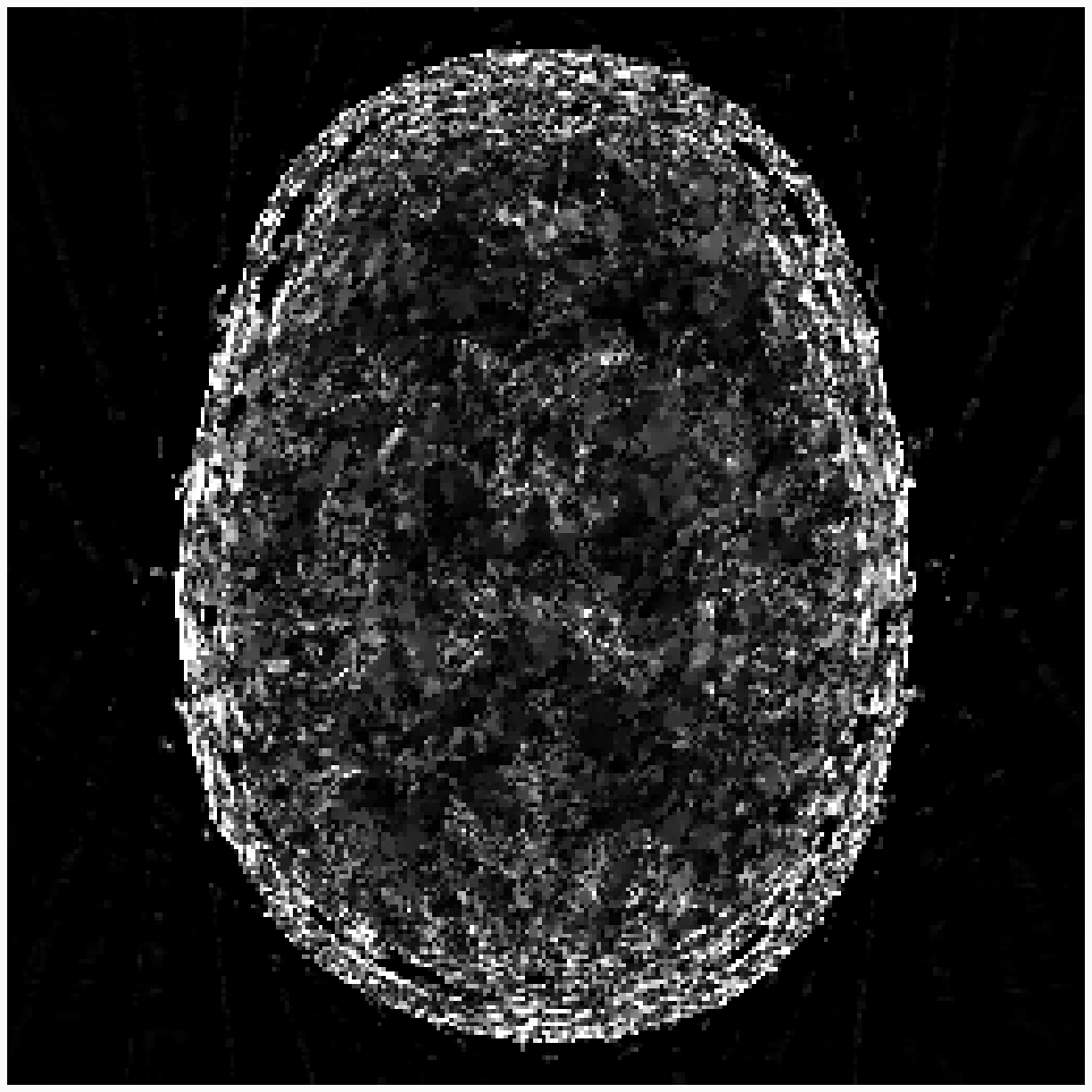}
        \caption{(SF - Ref)$\times 10^2$}
    \end{subfigure}
	\hskip2em
	\begin{subfigure}{.45\linewidth}
     	\includegraphics[scale=0.33]{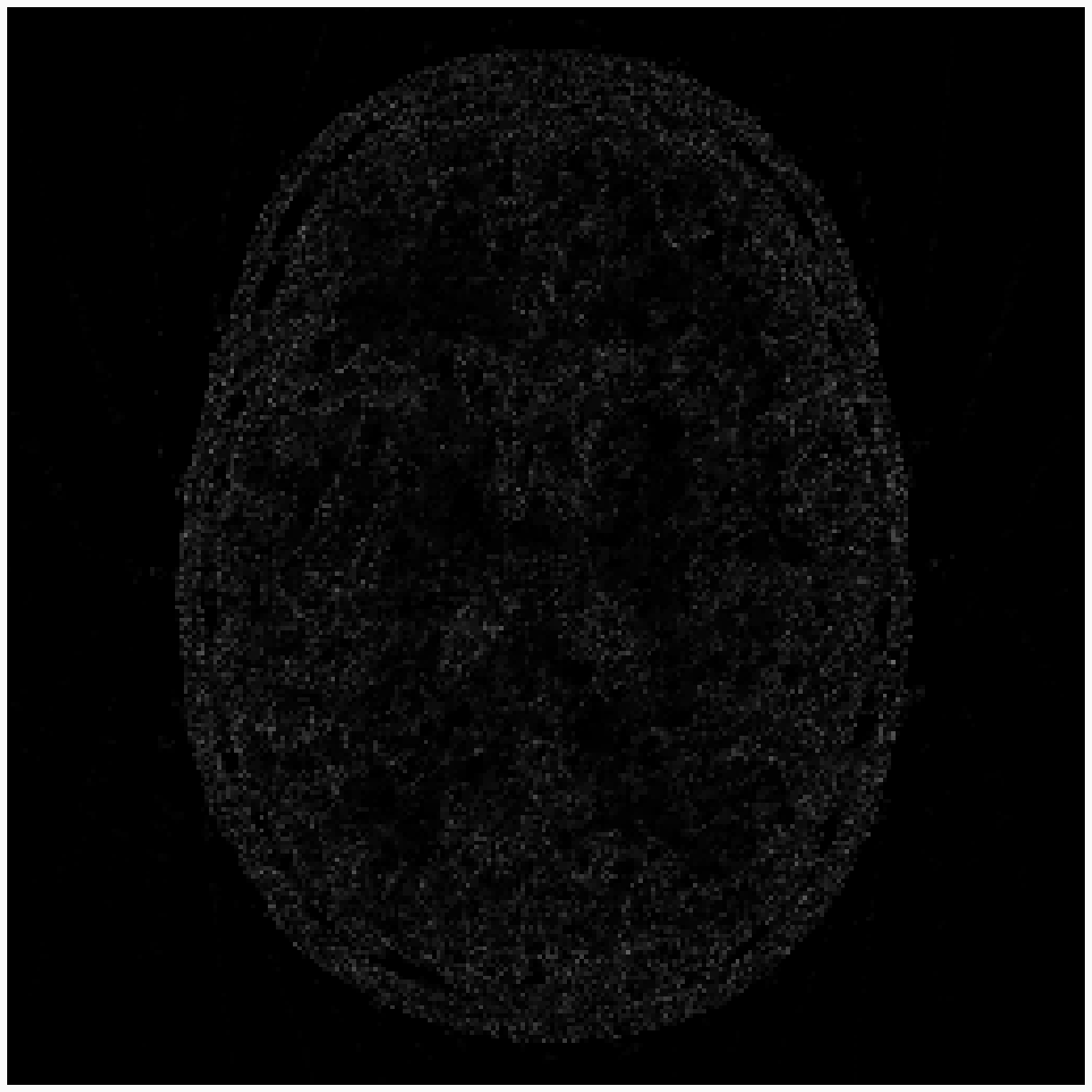}
        \caption{(CNSF - Ref)$\times (4\times10^3)$}
    \end{subfigure}
    \caption{Reconstruction of brain phantom from $30$ uniformly spaced projections using ASD-POCS. }\label{fig:reconstructionBrain}
\end{figure}
This experiment shows significant improvements over LTRI and SF methods in image reconstruction.

\section{Conclusion}
Accurate and efficient modeling the CT system is essential to the iterative image reconstruction problem.
We presented a convolutional non-separable footprint framework for forward and back-projection in fan-beam X-ray tomographic reconstruction. We show the detailed derivation from parallel X-ray transform to fan-beam setting.
The experiments, in a 2-D setting, show significant improvements in the approximation error of our method compared to other state-of-the-art methods designed for this purpose. The increase of the accuracy in forward model also results in the improvement of quality in image reconstruction. In addition, several times of speedup over the GPU and CPU implementations of other methods also shows the efficiency of our method. We believe that the implementation of evaluation of the CNSF is not fully optimized. Our future research will focus on this
optimization and meanwhile we are developing an extension to 3-D for cone-beam geometry.




\bibliographystyle{IEEEtran}
\bibliography{CTBoxspline}

\begin{thebibliography}{10}
\providecommand{\url}[1]{#1}
\csname url@samestyle\endcsname
\providecommand{\newblock}{\relax}
\providecommand{\bibinfo}[2]{#2}
\providecommand{\BIBentrySTDinterwordspacing}{\spaceskip=0pt\relax}
\providecommand{\BIBentryALTinterwordstretchfactor}{4}
\providecommand{\BIBentryALTinterwordspacing}{\spaceskip=\fontdimen2\font plus
\BIBentryALTinterwordstretchfactor\fontdimen3\font minus
  \fontdimen4\font\relax}
\providecommand{\BIBforeignlanguage}[2]{{%
\expandafter\ifx\csname l@#1\endcsname\relax
\typeout{** WARNING: IEEEtran.bst: No hyphenation pattern has been}%
\typeout{** loaded for the language `#1'. Using the pattern for}%
\typeout{** the default language instead.}%
\else
\language=\csname l@#1\endcsname
\fi
#2}}
\providecommand{\BIBdecl}{\relax}
\BIBdecl

\bibitem{pan2009commercial}
X.~Pan, E.~Y. Sidky, and M.~Vannier, ``Why do commercial ct scanners still
  employ traditional, filtered back-projection for image reconstruction?''
  \emph{Inverse problems}, vol.~25, no.~12, p. 123009, 2009.

\bibitem{gilbert1972iterative}
P.~Gilbert, ``Iterative methods for the three-dimensional reconstruction of an
  object from projections,'' \emph{Journal of theoretical biology}, vol.~36,
  no.~1, pp. 105--117, 1972.

\bibitem{andersen1984simultaneous}
A.~H. Andersen and A.~C. Kak, ``Simultaneous algebraic reconstruction technique
  (sart): a superior implementation of the art algorithm,'' \emph{Ultrasonic
  imaging}, vol.~6, no.~1, pp. 81--94, 1984.

\bibitem{andersen1989algebraic}
A.~H. Andersen, ``Algebraic reconstruction in ct from limited views,''
  \emph{IEEE transactions on medical imaging}, vol.~8, no.~1, pp. 50--55, 1989.

\bibitem{sidky2008image}
E.~Y. Sidky and X.~Pan, ``Image reconstruction in circular cone-beam computed
  tomography by constrained, total-variation minimization,'' \emph{Physics in
  Medicine \& Biology}, vol.~53, no.~17, p. 4777, 2008.

\bibitem{lewitt1990multidimensional}
R.~M. Lewitt, ``Multidimensional digital image representations using
  generalized kaiser--bessel window functions,'' \emph{JOSA A}, vol.~7, no.~10,
  pp. 1834--1846, 1990.

\bibitem{nilchian2015optimized}
M.~Nilchian, J.~P. Ward, C.~Vonesch, and M.~Unser, ``Optimized kaiser--bessel
  window functions for computed tomography,'' \emph{IEEE Transactions on Image
  Processing}, vol.~24, no.~11, pp. 3826--3833, 2015.

\bibitem{entezari2012box}
A.~Entezari, M.~Nilchian, and M.~Unser, ``A box spline calculus for the
  discretization of computed tomography reconstruction problems,'' \emph{IEEE
  transactions on medical imaging}, vol.~31, no.~8, pp. 1532--1541, 2012.

\bibitem{lo1988strip}
S.-C. Lo, ``Strip and line path integrals with a square pixel matrix: A unified
  theory for computational ct projections,'' \emph{IEEE transactions on medical
  imaging}, vol.~7, no.~4, pp. 355--363, 1988.

\bibitem{byonocore1981natural}
M.~H. Byonocore, W.~R. Brody, and A.~Macovski, ``A natural pixel decomposition
  for two-dimensional image reconstruction,'' \emph{IEEE Transactions on
  Biomedical Engineering}, no.~2, pp. 69--78, 1981.

\bibitem{tabei1992backprojection}
M.~Tabei and M.~Ueda, ``Backprojection by upsampled fourier series expansion
  and interpolated fft,'' \emph{IEEE transactions on image processing}, vol.~1,
  no.~1, pp. 77--87, 1992.

\bibitem{de2002distance}
B.~De~Man and S.~Basu, ``Distance-driven projection and backprojection,'' in
  \emph{2002 IEEE Nuclear Science Symposium Conference Record}, vol.~3.\hskip
  1em plus 0.5em minus 0.4em\relax IEEE, 2002, pp. 1477--1480.

\bibitem{long20103d}
Y.~Long, J.~A. Fessler, and J.~M. Balter, ``3d forward and back-projection for
  x-ray ct using separable footprints,'' \emph{IEEE transactions on medical
  imaging}, vol.~29, no.~11, pp. 1839--1850, 2010.

\bibitem{ha2018look}
S.~Ha and K.~Mueller, ``A look-up table-based ray integration framework for
  2-d/3-d forward and back projection in x-ray ct,'' \emph{IEEE transactions on
  medical imaging}, vol.~37, no.~2, pp. 361--371, 2018.

\bibitem{brahme2014comprehensive}
J.~A. Fessler., ``Comprehensive biomedical physics, chapter fundamentals of ct
  recon- struction in 2d and 3d,'' vol.~2, pp. 263--295, 2014.

\bibitem{rebollo2013sparse}
L.~Rebollo-Neira and J.~Bowley, ``Sparse representation of astronomical
  images,'' \emph{JOSA A}, vol.~30, no.~4, pp. 758--768, 2013.

\bibitem{natterer1986mathematics}
F.~Natterer, \emph{The mathematics of computerized tomography}.\hskip 1em plus
  0.5em minus 0.4em\relax Siam, 1986, vol.~32.

\bibitem{herman2009fundamentals}
G.~T. Herman, \emph{Fundamentals of computerized tomography: image
  reconstruction from projections}.\hskip 1em plus 0.5em minus 0.4em\relax
  Springer Science \& Business Media, 2009.

\bibitem{lu2009computable}
Y.~M. Lu, M.~N. Do, and R.~S. Laugesen, ``A computable fourier condition
  generating alias-free sampling lattices,'' \emph{IEEE Transactions on Signal
  Processing}, vol.~57, no.~5, pp. 1768--1782, 2009.

\bibitem{garcia2009study}
R.~Garc{\'\i}a-Moll{\'a}, R.~Linares, and R.~Ayala, ``A study of dqe dependence
  with beam quality, ge senographe essential detector for mammography,'' in
  \emph{World Congress on Medical Physics and Biomedical Engineering, September
  7-12, 2009, Munich, Germany}.\hskip 1em plus 0.5em minus 0.4em\relax
  Springer, 2009, pp. 886--889.

\bibitem{yu2012finite}
H.~Yu and G.~Wang, ``Finite detector based projection model for high spatial
  resolution,'' \emph{Journal of X-ray Science and Technology}, vol.~20, no.~2,
  pp. 229--238, 2012.

\bibitem{sutherland1974reentrant}
I.~E. Sutherland and G.~W. Hodgman, ``Reentrant polygon clipping,''
  \emph{Communications of the ACM}, vol.~17, no.~1, pp. 32--42, 1974.

\bibitem{deboor1993box}
C.~de~Boor, K.~H{\"o}llig, and S.~Riemenschneider, \emph{Box splines}.\hskip
  1em plus 0.5em minus 0.4em\relax Springer-Verlag, 1993, vol.~98.

\bibitem{Zhang2019box}
K.~Zhang and A.~Entezari, ``A convolutional forward and back-projection model
  for fan-beam geometry,'' in \emph{2019 IEEE International Symposium on
  Biomedical Imaging}.\hskip 1em plus 0.5em minus 0.4em\relax IEEE, 2019.

\bibitem{entezari2010box}
A.~Entezari and M.~Unser, ``A box spline calculus for computed tomography,'' in
  \emph{2010 IEEE International Symposium on Biomedical Imaging: From Nano to
  Macro}.\hskip 1em plus 0.5em minus 0.4em\relax IEEE, 2010, pp. 600--603.

\bibitem{irtFessler}
J.~Fessler, \emph{Michigan Image Reconstruction Toolbox (MIRT)}.\hskip 1em plus
  0.5em minus 0.4em\relax http://web.eecs.umich.edu/~fessler/irt/irt, 2014.

\bibitem{ltriHa}
K.~M. Sungsoo~Ha, \emph{Look-Up Table-Based Ray Integration Framework
  (LTRI)}.\hskip 1em plus 0.5em minus 0.4em\relax
  https://github.com/bsmind/LTRI, 2018.

\end{thebibliography}
\end{document}